\documentclass[12pt]{article}
\usepackage{amsfonts}
\usepackage{amsmath}
\usepackage{amssymb}
\usepackage{graphicx}
\usepackage{color}
\usepackage[all, knot]{xy}
\usepackage{tikz}
\usepackage{array}
\usepackage{hyperref}

\usepackage{comment}

\usepackage[utf8]{inputenc}
\usepackage{epstopdf}
\usepackage{subcaption}
\usepackage{amsthm}
\usepackage{enumitem}
\usepackage{mathrsfs}
\usepackage{mathtools}

\usepackage[margin=3cm]{geometry}

\def \be {\begin{equation}}
\def \ee {\end{equation}}
\def \bea {\begin{eqnarray}}
\def \eea {\end{eqnarray}}

\def \rr {\raise.35ex\hbox{\small $\prime$}\kern-.17em{\mbox{\large $\imath$}}}

\def \dels {\partial\kern-.6em /\kern.1em}
\def \As {{A\kern-.5em / \kern.5em}}
\def \Ds {D\kern-.7em / \kern.5em}

\def \ks {k\kern-.5em /}
\def \ls {l\kern-.5em /}

\newcommand{\ba}{\begin{eqnarray}}
\newcommand{\ea}{\end{eqnarray}}
\newcommand{\bal}{\begin{align}}
\newcommand{\eal}{\end{align}}
\newcommand{\bay}[1]{\left(\begin{array}{#1}}
\newcommand{\eay}{\end{array}\right)}

\newcommand{\ket}[1]{|{#1}\rangle}

\newcommand{\hide}[1]{}

\newcommand{\mT}{\mathrm{T}}

\newcommand{\md}{\mathrm{d}}
\newcommand{\supp}{\operatorname{supp}}
\newcommand{\Span}{\operatorname{Span}}
\newcommand{\Det}{\mathrm{Det}}

\DeclareMathOperator{\Tr}{Tr}
 
\DeclareMathOperator{\rank}{rank}

\newtheorem{proposition}{Proposition}

\newtheorem{lemma}{Lemma}

\newlist{axioms}{enumerate}{2}
\setlist[axioms,1]{label=\textbf{A\arabic{axiomsi}.}, ref=A\arabic{axiomsi}}
\setlist[axioms,2]{label=\textbf{A\arabic{axiomsi}\rlap{\myEnumCounter{axiomsii}}.},%
                   ref=A\arabic{axiomsi}\myEnumCounter{axiomsii},%
                   align=parleft,%
                   leftmargin=0em,%
                   itemsep=1.4ex,%
                   before={\stepcounter{axiomsi}}}

\usetikzlibrary{decorations.markings}

\begin{document}

\begin{titlepage}

\thispagestyle{empty}

\begin{center}
\textbf{\Large
Separability from Multipartite Measures
}\\
\vspace{1.5cm}
{\large
Chen-Te Ma$^{a,}$\footnote{e-mail address: \href{mailto:yefgst@gmail.com}{yefgst@gmail.com}} 
and
Ma-Ke Yuan$^{b,}$\footnote{e-mail address: \href{mailto:mkyuan19@fudan.edu.cn}{mkyuan19@fudan.edu.cn}} 
\\
\vspace{0.5cm}
}
{\it
$^a$Hetao Institute of  Mathematics and Interdisciplinary Sciences, Shenzhen,\\ 
Guangdong 518017, China.
\\
$^b$Department of Physics and Center for Field Theory and Particle Physics,\\ 
Fudan University, Shanghai 200433, China.
}\\
\end{center}
\vspace{0.5cm}
\begin{abstract}
\noindent
We show that multi-invariants provide necessary and sufficient criteria for full separability in finite-dimensional multipartite systems. 
For tripartite three-qubit states, the third-order negativity gives an exact criterion for full separability for both pure and mixed states. 
We then extend the framework to four-qubit pure states, whose full characterization requires eighteen independent real parameters. We identify eighteen independent entanglement measures, consisting of six bipartite, eight tripartite, and four quadripartite measures, with the third-order negativity serving as a key separability criterion. 
Finally, we generalize the construction to multipartite qudit systems and discuss an application to conformal field theory.
\end{abstract}
\end{titlepage}
\setcounter{tocdepth}{3}
{\hypersetup{linkcolor=black}\tableofcontents}
\newpage
\section{Introduction}
\label{sec:1}
\noindent
Understanding multipartite entanglement has become increasingly important in quantum field theory (QFT) and holography.
In quantum many-body systems, it offers a more refined probe of nonlocal correlations, topological order, and critical phenomena than bipartite entanglement measures such as entanglement entropy.
In the anti-de Sitter/conformal field theory (AdS/CFT) correspondence \cite{Maldacena:1997re,Gubser:1998bc,Witten:1998qj,Lunin:2000yv,Chang:2016ftb}, quantum entanglement plays a central role in relating boundary states to bulk geometry; multipartite measures and their holographic duals offer complementary probes of the entanglement structures underlying bulk spacetime~\cite{Gadde:2022cqi}.
Developing quantitative tools to characterize multipartite entanglement \cite{Harper:2024ker} is therefore essential for understanding both the structure of quantum many-body systems and the relation between entanglement and spacetime geometry.
\\

\noindent
Recently, multi-entropy~\cite{Gadde:2022cqi} and, more generally, multi-invariants~\cite{Gadde:2023zzj,Gadde:2024taa} have attracted considerable attention in high-energy theory.
For a normalized $\mathtt{q}$-partite pure state $|\Psi\rangle$, a general $N$-replica multi-invariant can be written as~\cite{Gadde:2023zzj,Gadde:2024taa}
\bea\label{eq:intro-multi-invariant}
\mathcal{Z}^{(\mathtt{q})}_{\Omega_{1,\dots,\mathtt{q}}}(\Psi)
\equiv
\langle\Psi|^{\otimes N}
\Omega_1\Omega_2\cdots\Omega_\mathtt{q}
|\Psi\rangle^{\otimes N}\ ,
\quad
\Omega_r\in S_N\ ,
\eea
where $\Omega_r$ permutes the $N$ replicas of the $r$-th subsystem.
R\'enyi multi-entropy corresponds to a particular permutation-symmetric family of such invariants and generalizes R\'enyi entropy to multipartite pure states.
In QFT, codimension-two twist operators or defects implement these replica permutations~\cite {Holzhey:1994we, Caraglio:2008pk, Ferrara:1973yt}.
In conformal field theories (CFTs), conformal symmetry determines the spacetime dependence of twist-operator correlation functions.
At the same time, theory-dependent data, such as operator product expansion (OPE) coefficients~\cite{Ferrara:1973yt}, can be obtained from the AdS/CFT correspondence.
For holographic CFTs, multi-entropy has been conjectured to be computed by the area of a bulk minimal web~\cite{Gadde:2022cqi}, generalizing the Ryu-Takayanagi formula to the multipartite case.
This framework has subsequently been extended to broader classes of multi-invariants~\cite{Gadde:2023zzj,Gadde:2024taa} and quantities designed to isolate genuine multipartite entanglement~\cite{Iizuka:2025ioc,Iizuka:2025caq}, suggesting that multi-invariants provide a natural probe for multipartite entanglement structures in CFT and holography.
\\

\noindent
Despite this progress, the precise quantum-information-theoretic content of multi-invariants remains less well understood than that of ordinary Rényi entropies.
See Refs.~\cite{Akella:2025owv, Yuan:2025dgx} for studies of multi-entropy of stabilizer states and Refs.~\cite{Gadde:2025csh, Gadde:2025ybn} for constructing entanglement monotones from multi-invariants.
A natural question is whether multi-invariants can do more than characterize or quantify multipartite correlations: in particular, whether their extremal properties can provide exact criteria for the \textit{separability} of quantum systems.
\\

\noindent
In quantum information theory, deciding whether a density matrix is separable, namely, whether it can be written as
\bea
\rho=\sum_jp_j\rho_{1}^{(j)}\otimes\rho_{2}^{(j)}\otimes\cdots\otimes\rho_{n}^{(j)} \ ,
\eea
is known as the {\it separability} problem.
Although separable states may contain correlations, these correlations are purely {\it classical} rather than quantum.
In general, the separability problem is believed to be computationally hard.
\\

\noindent
A useful strategy is therefore to construct entanglement measures that vanish for separable states, providing a separability criterion.
The simplest example is the bipartite two-qubit pure state, for which the {\it two-tangle}
\begin{equation}\label{eq:2-tangle-intro}
\tau_{A|B} = 2\big(1 - \Tr (\rho_A^2)\big)
\end{equation}
vanishes if and only if the state is a product state~\cite{Bennett:1996gf,Hill:1997pfa,Wootters:1997id}.
Its key ingredient is the reduced purity $\Tr(\rho_A^2)$, which reaches its maximum value $1$ precisely for {\it product} states and is {\it symmetric} under the exchange of $A$ and $B$.
The {\it convex-roof} extension,
\begin{equation}\label{eq:2-tangle-cr}
E_{\tau_{A|B}}(\rho)\equiv
\inf_{\{p_\ell, \ket{\psi_\ell}\}\in\mathbb{D(\rho)}}
\sum_\ell p_\ell\, \tau_{A|B}(\psi_\ell) \ ,
\end{equation}
with $\mathbb{D}(\rho)$ the set of all convex decompositions of $\rho$ into pure states, then provides a natural mixed-state separability criterion~\cite{Bennett:1996gf,Hill:1997pfa,Wootters:1997id}.
Importantly, the reduced purity $\Tr(\rho_A^2)$ is itself the simplest multi-invariant with two replicas, i.e., Eq.~\eqref{eq:intro-multi-invariant} with $N = 2$, $\Omega_A=(12)$, and $\Omega_B=\mathrm{id}$.
Its role in the two-tangle, therefore, suggests that more general multi-invariants~\eqref{eq:intro-multi-invariant} may provide natural multipartite generalizations of reduced purity and can likewise furnish exact separability criteria.
One of the main results of this work is that multi-invariants indeed provide such generalizations: just as the reduced purity reaches its maximal value $1$ precisely for a bipartite product pure state, the multi-invariants studied below are bounded above by $1$ and reach this value if and only if the state is fully product.
This extremality property lets multi-invariants construct exact criteria for full separability and gives them a quantum-information-theoretic interpretation.
\\

\noindent
\textbf{Three-qubit pure states.} The simplest nontrivial realization of this idea is provided by three-qubit pure states.
Unlike a two-qubit pure state, which is characterized by a single independent Schmidt coefficient, a generic three-qubit pure state requires five continuous parameters. Explicitly, it is LU-equivalent to the canonical form~\cite{Peres:1994qv,Acin:2000jx,Brun:2001gyy}
\begin{equation}\label{eq:Acin}
|\Psi\rangle
=
\lambda_0 |000\rangle
+\lambda_1 e^{i\varphi}|100\rangle
+\lambda_2 |101\rangle
+\lambda_3 |110\rangle
+\lambda_4 |111\rangle \ ,
\end{equation}
with $\lambda_j\ge 0$, $0\le \varphi <\pi$, and $\sum_{j=0}^4 \lambda_j^2=1$.
This already shows the increased structural complexity of {\it tripartite} entanglement. Beyond bipartite measures, two important genuinely tripartite quantities arise~\cite{Sudbery:2001jxc}:
\begin{itemize}
\item{{\it Three-Tangle}~\cite{Coffman:1999jd}
\begin{equation}\label{eq:CKW}
\tau_{ABC} = \tau_{A|BC} - \tau_{A|B} - \tau_{A|C} \ ,
\end{equation}
which captures the residual entanglement that cannot be reduced to pairwise correlations.
Here $\tau_{A|BC} = 2(1 - \Tr(\rho_A^2))$, while $\tau_{A|B}$ and $\tau_{A|C}$ are defined via the convex-roof extension for mixed states, i.e., Eq.~\eqref{eq:2-tangle-cr}.
}
\item{{\it Third-Order Negativity}~\cite{Sudbery:2001jxc}
\begin{equation}\label{eq:I5-intro}
I_5 = \Tr\!\big[(\rho_{BC}^{\Gamma})^3\big] \ ,
\quad
\rho_{BC} = \Tr_A |\Psi\rangle\langle\Psi| \ ,
\end{equation}
where $\Gamma$ denotes the partial transpose.
Although $\rho_{BC}^{\Gamma}$ is {\it not} positive in general, it is Hermitian, so $I_5$ is always real.
For a tripartite pure state, $I_5$ admits a three-replica representation, i.e., Eq.~\eqref{eq:intro-multi-invariant} with $N = 3$, $\Omega_A=\mathrm{id}$, $\Omega_B=(123)$, and $\Omega_C=(132)$.
Thus, $I_5$ is a tripartite multi-invariant.
We show that $I_5 \le 1$ and that it reaches its maximum value if and only if the state is fully product.
It therefore plays, for tripartite pure states, a role analogous to that of the reduced purity in the bipartite case.}
\end{itemize}

\noindent
The positive partial transpose (PPT) criterion provides a standard diagnostic of bipartite entanglement: a negative eigenvalue of the partially transposed density matrix necessarily implies entanglement~\cite{Peres:1996dw}.
This observation underlies the negativity~\cite{Vidal:2002zz} and logarithmic negativity, which are
widely used entanglement measures~\cite{Vedral:1997qn,Vidal:1998re,Nielsen:1999zza}.
However, PPT is necessary and sufficient for separability only in $2\times2$ and $2\times3$ Hilbert spaces~\cite{Horodecki:1996nc}; in higher dimensions, PPT-entangled states exist.
Moreover, pairwise negativities do not provide a unified criterion for the full separability of a multipartite state.
By contrast, the single permutation-symmetric invariant $I_5$ detects every non-fully-product three-qubit pure state through the strict inequality $I_5<1$.
Its saturation therefore yields a full-separability criterion that cannot be obtained solely from the negativities of the pairwise reductions.
\\

\noindent
The third-order negativity $I_5$ has two further notable advantages: first, it is experimentally \textit{accessible}, for example via randomized measurements or tomography-inspired protocols; second, for three-qubit pure states, saturation of its lower bound identifies the W-state~\cite{Osterloh:2008ddl}, providing a way to distinguish entanglement classes. However, a key drawback is that $1 - I_5$ is not by itself an entanglement monotone.
To address this issue, Ref.~\cite{Oreshkov:2006gqe} introduced the quantity
\begin{equation}\label{eq:phiABC}
\phi_{ABC} = 69 - \Tr\big[(2\rho_{AB} + \rho_A\otimes I_B + I_A\otimes \rho_B)^3\big] - 3\Tr(\rho_{AB}^2)\ ,
\end{equation}
which is an entanglement measure.
Using identities specific to three-qubit pure states, we show that this quantity can be rewritten in terms of $I_5$ and the tangles:
\begin{equation}\label{eq:phiABC-new}
\phi_{ABC} = 12(1 - I_5) + 27(\tau_{A|B} + \tau_{A|C} + \tau_{B|C}) + \frac{81}{2}\tau_{ABC} \ .
\end{equation}
This relation makes the permutation symmetry of $\phi_{ABC}$ manifest.
It shows how the multi-invariant $I_5$ combines with conventional bipartite and tripartite tangles to form an entanglement monotone.
Since $\phi_{ABC}$ vanishes precisely on fully product pure states, its convex-roof extension defines an entanglement monotone for mixed states.
It vanishes if and only if the mixed state is fully separable, thereby providing a criterion for full separability.
\\

\noindent
\textbf{Four-qubit pure states/three-qubit mixed states.} Four-qubit pure states provide the simplest setting in which the tripartite mixed-state criterion developed above becomes directly applicable: tracing out any one qubit produces a three-qubit mixed state.
A $\mathtt{q}$-qubit pure state has $(2^{\mathtt{q}+1}-2-3\mathtt{q})$ independent real parameters (after removing normalization, global phase, and LU freedoms) when $\mathtt{q}>2$~\cite{Carteret:2000jop}.
In particular, this gives $5$ for $\mathtt{q}=3$ and $18$ for $\mathtt{q}=4$.
This motivates us to explore four-qubit entanglement through a collection of eighteen measures associated with different numbers of parties.
We can organize these 18 quantities into 6 bipartite measures (two-tangles), 4 tripartite measures (three-tangles), 4 ``$\phi$-type'' measures, and 4 quadripartite measures.
All these 18 quantities are expected to satisfy
\begin{equation}
E\Big(\sum_j p_j \rho_j\Big)\le \sum_j p_j E(\rho_j)\le E(\rho)\ ,
\end{equation}
so that they remain valid entanglement measures under convex-roof extension~\cite{Eisert:2001xyp}.
This guarantees that these pure-state measures admit consistent extensions to mixed states.
\\

\noindent
To figure out the 4 quadripartite measures, we first note that the four-tangle~\cite{Uhlmann:2000ckp,Wong:2000cmz}
\begin{equation}\label{eq:4-tangle}
4|H|^2 = C^2_{1234}(\Psi)\ ,\quad\text{with}\ \ C_{1234}(\Psi) = |\langle\Psi^*|\sigma_y^{\otimes 4}|\Psi\rangle|
\end{equation}
plays the role of a genuine quadripartite invariant, analogous to the two-qubit concurrence (two-tangle) $\tau_{A|B} = C^2(\Psi)$ with $C(\Psi) = |\langle\Psi|\sigma_y^{\otimes 2}|\Psi^*\rangle|$.
The key structural identity~\cite{Gour:2010dhm}
\begin{equation}\label{eq:H2-tau}
4|H|^2 = \tau_{A|BCD} + \tau_{B|ACD} + \tau_{C|ABD} + \tau_{D|ABC} - \tau_{AB|CD} - \tau_{AC|BD} - \tau_{AD|BC}
\end{equation}
shows that $4|H|^2$ is not determined solely by lower-order bipartite and tripartite correlations, but instead provides an independent quadripartite contribution.
It is therefore a natural ingredient in the eighteen-parameter set of entanglement measures.
For a positive homogeneous function $\mu$ defined on pure multipartite qubit states $\Psi$, invariant under local determinant-one SLOCC operations, and satisfying
\bea
\mu(\lambda\Psi)=\lambda^{\eta}\mu(\Psi)\ ,
\eea
the condition that the degree $\eta\le 4$ ensures that $\mu$ is an entanglement measure~\cite{Eltschka:2012voz}. Since $4|H|^2$ is degree-4, it is an entanglement measure.
Another useful observation is that both the two-tangle and the three-tangle can be rewritten in terms of a hyperdeterminant.
Therefore, another natural candidate is the four-qubit hyperdeterminant $\Delta$, which has degree 24.
Even though it is complicated, the key insight is that it can be expressed in terms of {\it lower}-degree invariants $H$, $W$, $\Sigma$, and $\Pi$, of degrees $2$, $6$, $8$, and $12$, respectively~\cite{Tapia:2002du,Luque:2003dun,Miyake:2002vut}. 
Taking suitable absolute-value powers makes them degree-4 monotones.
See Table~\ref{tab:hierarchy-measures} for a summary.
\\

\noindent
Different entanglement measures probe different aspects of correlations \cite{Dur:2000zz,Verstraete:2002gqj}.
Two states may both be entangled yet remain incomparable under a single scalar measure.
Therefore, it is natural to ask whether any one of the eighteen entanglement measures can exist independently, i.e., remain nonzero while all the others vanish.
The four $\phi$-type measures connect this question to mixed-state separability: their simultaneous vanishing implies that all three-qubit reduced density matrices are fully separable. By Ref.~\cite{Thapliyal:1998nw}, the global state must then be LU-equivalent to a generalized GHZ-state,
\bea\label{eq:intro-GHZ}
|\widetilde{\mathrm{GHZ}}\rangle
=
a|0000\rangle+b|1111\rangle \ .
\eea
For this state, the four-tangle $4|H|^2$ is the only nonvanishing measure. We further show that $4|H|^2$ is the unique measure that can exist independently. Thus, the generalized GHZ-state is the only state whose entanglement can be quantified unambiguously.
\\

\noindent
\textbf{Multipartite qudit systems.}
The above results extend naturally to general finite-dimensional multipartite systems.
For any normalized pure state, we show that the norm of a multi-invariant is bounded from above by 1, and fully product states saturate this bound.
These criteria also impose strong constraints on the global entanglement structure: if every $(\mathtt{q} - 1)$-qudit reduced density matrix of a $\mathtt{q}$-qudit pure state is fully separable, then the state is LU-equivalent to a generalized GHZ-state~\cite{Thapliyal:1998nw}.
We can also use qudit systems as a controlled bridge to QFT \cite{Witten:2018zxz,Balasubramanian:2025hxg}.
This works because qudit systems mimic QFT subregions with large local dimension.
Therefore, they provide a clean laboratory for multipartite structures~\cite{Gadde:2022cqi,Gadde:2023zzj} and separability criteria.

\subsection{Summary of Results}
\noindent
This paper explores the relationship between separability and multi-invariants.
Specifically:
\begin{table}[tpb]
\centering
\renewcommand{\arraystretch}{1.25}
\begin{tabular}{|c|c|c|c|}
\hline

& 2-qubit pure state
& 3-qubit pure state
& 4-qubit pure state \\
\hline
Bipartite
& $\tau_{A|B}$
& $\tau_{A|B},\,\tau_{A|C},\,\tau_{B|C}$
& \begin{tabular}{@{}l@{}}
$\tau_{A|B},\,\tau_{A|C},\,\tau_{A|D}$ \\[-0.2em]
$\tau_{B|C},\,\tau_{B|D},\,\tau_{C|D}$
\end{tabular}
\\
\hline
Tripartite
& --
& $\tau_{ABC},\,\phi_{ABC}$
& \begin{tabular}{@{}l@{}}
$\tau_{ABC},\,\tau_{ABD},\,\tau_{ACD},\,\tau_{BCD}$\\[-0.2em]
$\phi_{ABC},\,\phi_{ABD},\,\phi_{ACD},\,\phi_{BCD}$
\end{tabular}
\\
\hline
Quadripartite
& --
& --
& $4|H|^2,\,|\Sigma|^{1/2},\,|\Pi|^{1/3},\,|\Delta|^{1/6}$
\\
\hline
\end{tabular}
\caption{Hierarchical organization of entanglement measures for pure states of two, three, and four qubits.}
\label{tab:hierarchy-measures}
\end{table}
\begin{itemize}
\item
We show that the third-order negativity $I_5$ attains its maximal value~$1$ if and only if the state is fully product~\cite{Carrozza:2026qcf}.
Motivated by the convex-roof construction in the bipartite case, we develop a corresponding mixed-state extension and obtain a necessary and sufficient criterion for full separability of tripartite mixed states.
In this sense, $\phi_{ABC}$~\eqref{eq:phiABC} provides a natural tripartite analog of the two-tangle~\eqref{eq:2-tangle-intro}, while $I_5$ plays a role similar to the reduced purity $\Tr(\rho_A^2)$ in Eq.~\eqref{eq:2-tangle-intro}.
We also prove that, for three-qubit pure states, the lower bound on $I_5$ and the upper bound on $\phi_{ABC}$ are saturated uniquely by the W- and GHZ-states, respectively.
\item
We identify a set of eighteen independent entanglement measures for four-qubit pure states, as summarized in Table~\ref{tab:hierarchy-measures}.
Among them, $\phi_{ABC}$ and its permutations play a distinguished role: they are closely related to $I_5$ and provide necessary and sufficient criteria for full separability of tripartite states.
Furthermore, we show that the remaining quadripartite measures— $4|H|^2$, $|\Sigma|^{1/2}$, and $|\Pi|^{1/3}$—are all generated by the hyperdeterminant $\Delta$, motivating the search for new entanglement measures derived from $\Delta$.
\item
We demonstrate that among the eighteen measures, the four-tangle $4|H|^2$ is the only one that can exist independently.
This situation is realized uniquely by the four-qubit generalized GHZ-type states.
Consequently, the entanglement of generalized GHZ-states can be quantified unambiguously.
\item
We generalize the relation between separability and third-order negativity to arbitrary finite-dimensional multipartite systems, using multi-invariants for pure states~\cite{Carrozza:2026qcf} and mixed states.
We provide an alternative proof of the result in Ref.~\cite{Thapliyal:1998nw}: the condition that all $(\mathtt{q}-1)$-partite reduced density matrices are fully separable is equivalent to the statement that the $\mathtt{q}$-partite pure state is LU-equivalent to a generalized GHZ-state.
This directly applies multipartite separability criteria to characterize the global entanglement structure and implies that the GHZ-type state is not allowed in the local QFT for the spatial bipartition.
\item
Finally, we demonstrate an application in QFT by computing the third-order negativity in two-dimensional conformal field theory (CFT$_2$) in the large-central-charge limit, using the AdS/CFT correspondence to extract the relevant OPE coefficients.
\end{itemize}

\noindent
The paper is organized as follows. In Sec.~\ref{sec:separability}, we establish the relation between third-order negativity and tripartite separability.
In Sec.~\ref{sec:18}, we introduce the set of eighteen entanglement measures, analyze their structural properties, and show that the four-tangle is the unique independently surviving measure.
In Sec.~\ref{sec:multipartite}, we generalize our results to multipartite systems.
In Sec.~\ref{sec:3rd-nega-CFT}, we present the CFT application of third-order negativity.
Finally, Sec.~\ref{sec:discussion} summarizes our results and discusses future directions.
Appendices~\ref{appx:I5-lower} and \ref{appx:phi-upper} provide detailed proofs of the lower bound on $I_5$ and the upper bound on $\phi_{ABC}$. 
Appendix~\ref{appx:OPE} derives the twist-operator OPE coefficient used in Sec.~\ref{sec:3rd-nega-CFT}.
\section{Separability Criteria of Tripartite Systems}
\label{sec:separability}
\noindent
In this section, we show that the third-order negativity provides a criterion for full separability of tripartite systems.
We begin with pure states, where the result takes a simple form: for a finite-dimensional tripartite pure state, $I_5=1$ if and only if the state is fully product.
We then turn to three-qubit mixed states and obtain the criterion that $E_{\phi_{ABC}}=0$ if and only if the state is fully separable, where $\phi_{ABC}$ is a quantity closely related to $I_5$.
Finally, based on $I_5$, we develop another mixed-state extension $E_{1-I_5}$, whose vanishing serves as a criterion for full separability in general finite-dimensional tripartite mixed states.
We also show in Appendices~\ref{appx:I5-lower} and~\ref{appx:phi-upper} that for three-qubit pure states, the lower bound of $I_5$ is $2/9$, achieved by the W-state, and the upper bound of $\phi_{ABC}$ is $99/2$, achieved by the GHZ-state.
Hence, one can determine how close a three-qubit pure state is to the W- and GHZ-states by measuring $I_5$ and $\phi_{ABC}$.
\subsection{Tripartite Pure States}
\label{sec:I5-prop}
\noindent
For tripartite pure states $|\Psi\rangle_{ABC}$, the third-order negativity $I_5$~\eqref{eq:I5-intro} can also be expressed through the replica trick as
\begin{equation}\label{eq:I5-replica}
I_5
=
\langle \Psi|^{\otimes 3}
\Omega_A \Omega_B \Omega_C
|\Psi\rangle^{\otimes 3} \ ,\ \ \text{with}\ \
\Omega_A = (1)(2)(3)\ ,
\
\Omega_B = (123) \ ,
\
\Omega_C = (132) \ .
\end{equation}
Here $\Omega_A$, $\Omega_B$, and $\Omega_C$ are permutation operators acting on the replica copies of the corresponding subsystems.
In particular, $\Omega_A=(1)(2)(3)$ is the identity permutation on the three replica copies of subsystem $A$, while $\Omega_B=(123)$ cyclically permutes the three replicas of subsystem $B$, i.e.,
\bea
\Omega_B(1)=2\ ,\quad \Omega_B(2)=3\ ,\quad \Omega_B(3)=1\ ,
\eea
whereas $\Omega_C$ acts in the opposite direction.
From Eq.~\eqref{eq:I5-replica} one can see that $I_5$ is symmetric among $A$, $B$, and $C$.
We now prove the following proposition:
\begin{proposition}\label{prop:I5-product}
Let $|\Psi\rangle \in \mathcal{H}_A \otimes \mathcal{H}_B \otimes \mathcal{H}_C$ be a normalized 3-qudit pure state, then
\begin{equation}
I_5 = 1 \iff |\Psi\rangle \text{ is a fully product state}\ .
\end{equation}
\end{proposition}
\begin{proof}
The converse is immediate. For the nontrivial direction, we first prove that for a general bipartite mixed state $\rho_{BC}$,
\begin{equation}\label{eq:I5s1}
I_5
=
\Tr\big[(\rho_{BC}^{\Gamma})^3\big] \le 1\ .
\end{equation}
Let $X = \rho_{BC}^{\Gamma}$.
Since $X$ is Hermitian, its eigenvalues $\mu_j$ are real, and
\begin{equation}
I_5 = \sum_j \mu_j^3 \ .
\end{equation}
We now show that $\mu_j \le 1$ for all $j$.
It is clear that
\begin{equation}
\mu_j\le\lambda_{\max}(X)
=
\max_{\|v\|=1} \langle v|X|v\rangle
=
\max_{\|v\|=1}
\Tr\!\big[\rho_{BC}(|v\rangle\langle v|)^{\Gamma}\big]
\le
\max_{\|v\|=1}
\lambda_{\max}\big((|v\rangle\langle v|)^{\Gamma}\big) \ .
\end{equation}
Now let $|v\rangle = \sum_k s_k |k\rangle_B |k\rangle_C$ be the Schmidt decomposition of $|v\rangle$.
Then the eigenvalues of $(|v\rangle\langle v|)^{\Gamma}$ are $s_k^2$ and $\pm s_k s_l$ with $k<l$.
Hence, we obtain
\begin{equation}
\lambda_{\max}\big((|v\rangle\langle v|)^{\Gamma}\big)
=
\max_k s_k^2
\le 1 \Longrightarrow \lambda_{\max}(X)\le 1 \ ,
\end{equation}
so every eigenvalue $\mu_j$ of $X$ satisfies $\mu_j \le 1$.
It follows that $\mu_j^3 \le \mu_j^2$ for every $j$, and therefore
\begin{equation}
I_5
=
\sum_j \mu_j^3
\le
\sum_j \mu_j^2
=
\Tr(X^2)
=
\Tr(\rho_{BC}^2)
\le 1 \ ,
\end{equation}
where the third equal sign follows from the fact that the partial transpose preserves the Hilbert-Schmidt norm.
Now suppose that $I_5 = 1$.
Since
\begin{equation}
I_5 \le \Tr(\rho_{BC}^2) \le 1 \ ,
\end{equation}
both inequalities must be saturated.
In particular, $\Tr(\rho_{BC}^2) = 1$ implies that $\rho_{BC}$ is pure.
Hence, $|\Psi\rangle$ factorizes across $A|BC$,
\begin{equation}
|\Psi\rangle = |A\rangle \otimes |\phi\rangle_{BC} \ .
\end{equation}
For a bipartite pure state $|\phi\rangle_{BC}$, one has
\begin{equation}
\Tr\!\big[\big((|\phi\rangle\langle\phi|)^{\Gamma}\big)^3\big]
=
\Tr(\rho_B^3)\ ,
\quad
\rho_B = \Tr_C |\phi\rangle\langle\phi| \ .
\end{equation}
Therefore, we have $1=I_5=\Tr(\rho_B^3)$, which implies that $\rho_B$ is also pure.
Thus $|\phi\rangle_{BC}$ is product across $B|C$, and $|\Psi\rangle$ is fully product across $A|B|C$.
\end{proof}
\noindent
\boldmath \textbf{Lower bound of $I_5$ for three-qubit pure states.}\unboldmath\ The lower bound is $I_5 \ge 2/9$, which is saturated when the state is a W-state~\cite{Osterloh:2008ddl}.
See Appendix~\ref{appx:I5-lower} for a detailed proof.
Therefore, measuring $I_5$ determines how close a 3-qubit pure state is to the W-state.
\subsection{Three-Qubit Mixed States}
\label{sec:phi-prop}
\noindent
We first derive Eq.~\eqref{eq:phiABC-new}, which rewrites $\phi_{ABC}$~\eqref{eq:phiABC} in terms of $ I_5$.
It has been shown in Ref.~\cite{Oreshkov:2006gqe} that Eq.~\eqref{eq:phiABC} can be expanded as
\begin{equation}\label{eq:phi-expand}
\phi_{ABC} = 69 - 12I_5 - 16\big[\!\Tr(\rho_A^3) + \Tr(\rho_B^3) + \Tr(\rho_C^3)\big] - 3\Tr(\rho_A^2) - 3\Tr(\rho_B^2) - 3\Tr(\rho_{AB}^2)\ .
\end{equation}
For three-qubit pure states, $\Tr(\rho_{AB}^2) = \Tr(\rho_C^2)$. 
For a normalized density matrix $\rho$ with rank no more than 2, we have $\Tr(\rho^3) = (3\Tr(\rho^2) - 1)/2$. Therefore,
\begin{equation}\label{eq:phiABC-m}
\phi_{ABC} = 93 - 12I_5 - 27\big[\!\Tr(\rho_A^2) + \Tr(\rho_B^2) + \Tr(\rho_C^2)\big]\ .
\end{equation}
Using the definition of the two-tangle of pure states and Eq.~\eqref{eq:CKW}, we have
\begin{equation}\label{eq:purity3}
\begin{split}
\Tr(\rho_A^2) + \Tr(\rho_B^2) + \Tr(\rho_C^2) &= 3 - \frac{1}{2}(\tau_{A|BC} + \tau_{B|CA} + \tau_{C|AB})\\
&= 3 - (\tau_{A|B} + \tau_{A|C} + \tau_{B|C}) - \frac{3}{2}\tau_{ABC}\ .
\end{split}
\end{equation}
Substituting Eq.~\eqref{eq:purity3} into Eq.~\eqref{eq:phiABC-m}, we obtain Eq.~\eqref{eq:phiABC-new}.
\\

\noindent
We then show that $\phi_{ABC}$ can serve as a necessary and sufficient criterion for full separability of three-qubit mixed states.
\begin{proposition}
Let
\begin{equation}\label{eq:mixed-state-phi}
E_{\phi_{ABC}}(\rho)\equiv
\inf_{\{p_\ell,\ket{\psi_\ell}\}\in \mathbb{D}(\rho)}
\sum_\ell p_\ell\, \phi_{ABC}(\psi_\ell) \ ,
\end{equation}
where $\phi_{ABC}$ is defined as Eq.~\eqref{eq:phiABC-new}:
\begin{equation*}
\phi_{ABC} = 12(1 - I_5) + 27(\tau_{A|B} + \tau_{A|C} + \tau_{B|C}) + \frac{81}{2}\tau_{ABC} \ .
\end{equation*}
Then, for any three-qubit mixed state $\rho$,
\begin{equation}
E_{\phi_{ABC}}(\rho)=0
\iff
\rho \ \text{is fully separable} \ .
\end{equation}
\end{proposition}
\begin{proof}
On the one hand, since $\tau_{A|B}$, $\tau_{B|C}$, $\tau_{C|A}$, $\tau_{ABC}$ and $1 - I_5$ are all nonnegative, $\phi_{ABC} = 0$ implies $1 - I_5 = 0$.
According to Proposition~\ref{prop:I5-product}, since $1-I_5$ is a continuous nonnegative function on the pure-state manifold, there exists a pure-state decomposition $\rho=\sum_j p_j |\psi_j\rangle\langle\psi_j|$ such that
\begin{equation}\label{eq:I5-0}
\sum_j p_j \bigl(1-I_5(\psi_j)\bigr)=0 \ .
\end{equation}
Since each term in Eq.~\eqref{eq:I5-0} is nonnegative, it follows that
\begin{equation}
1-I_5(\psi_j)=0
\end{equation}
for every $j$ with $p_j>0$.
Hence, every $\ket{\psi_j}$ is fully product, and thus $\rho$ is fully separable.
On the other hand, if $\rho$ is fully separable, then there exists a pure-state decomposition $\rho=\sum_j p_j |\psi_j\rangle\langle\psi_j|$ such that every $\ket{\psi_j}$ is fully product.
For every fully product state $\ket{\psi_j}$, we have
\begin{equation}
\tau_{A|B} = \tau_{B|C} = \tau_{C|A} = \tau_{ABC} = 1 - I_5 = 0\ ,
\end{equation}
thus $E_{\phi_{ABC}}(\rho) = 0$.
\end{proof}
\noindent
\boldmath \textbf{Upper bound of $\phi_{ABC}$ for three-qubit pure states.}\unboldmath\ The upper bound of $\phi_{ABC}$ is $\phi_{ABC} \le 99/2$, which is saturated by the GHZ-state. See Appendix~\ref{appx:phi-upper} for a proof. 
Hence, $\phi_{ABC}$ also indicates how close a 3-qubit pure state is to the GHZ-state.

\subsection{Tripartite Mixed States}
\label{sec:mixed-3}
\noindent
The measure $\phi_{ABC}$ is defined for three-qubit states.
For general tripartite mixed states, we can construct a criterion for their full separability based on $I_5$ and the convex-roof construction~\cite{Eisert:2001xyp} as follows.
\begin{proposition}
Let
\begin{equation}\label{eq:1mI5}
E_{1 - I_5}(\rho)\equiv
\inf_{\{p_\ell,\ket{\psi_\ell}\}\in \mathbb{D}(\rho)}
\sum_\ell p_\ell\, \big(1 - I_5(\psi_\ell)\big) \ .
\end{equation}
Then, for any three-qudit mixed state $\rho$,
\begin{equation}
E_{1 - I_5}(\rho)=0
\iff
\rho \ \text{is fully separable} \ .
\end{equation}
\end{proposition}
\begin{proof}
This follows directly from Eq.~\eqref{eq:I5s1}, which gives $I_5\leq 1$ and hence $1-I_5\geq 0$.
\end{proof}
\section{Entanglement Measures for Four-Qubit Pure States}
\label{sec:18}
\noindent
Four-qubit pure states provide the simplest but nontrivial setting where the criteria in Sec.~\ref{sec:separability} become directly useful: tracing out one qubit yields a tripartite mixed state.
At the same time, the pure-state system itself already exhibits a rich multipartite structure.
A four-qubit pure state modulo local unitary transformations is characterized by eighteen independent real parameters~\cite{Carteret:2000jop}. 
Motivated by this parameter counting, we identify eighteen independent entanglement measures, consisting of six bipartite, eight tripartite, and four quadripartite measures, as summarized in Table~\ref{tab:hierarchy-measures}.
We now introduce these measures and discuss their properties.
In particular, we explain how the hyperdeterminant generates these entanglement measures.
Finally, we evaluate these entanglement measures on representative states and list the 9 ways in which 4-qubit pure states can be entangled.
\subsection{Bipartite Entanglement Measures}
\label{subsec:bi}
\noindent
The two-tangle~\eqref{eq:2-tangle-intro} is an entanglement measure originally defined for two-qubit pure states~\cite{Bennett:1996gf,Hill:1997pfa,Wootters:1997id}
\begin{equation}
|\Psi\rangle = \sum_{j,k = 0}^1 t_{jk}|jk\rangle\ .
\end{equation}
The definition~\eqref{eq:2-tangle-intro} generalizes to finite-dimensional bipartite pure states.
If we limit to two-qubit pure states, the two-tangle can also be expressed in alternative forms, such as the square of the two-qubit concurrence~\cite{Bennett:1996gf,Hill:1997pfa,Wootters:1997id}
\begin{equation}
\tau_{A|B} = C^2(\Psi) \ ,\quad \text{with}\ \ C(\Psi) \equiv |\langle\Psi|\sigma_y^{\otimes 2}|\Psi^*\rangle|\ ,
\end{equation}
the determinant of the reduced density matrix~\cite{Bennett:1996gf,Hill:1997pfa,Wootters:1997id}
\begin{equation}
\tau_{A|B} = 4\det \rho_A\ ,
\end{equation}
and the hyperdeterminant of the coefficient tensor
\begin{equation}\label{eq:2hd}
\tau_{A|B} = 4|\Det\, \mT_2|^2\ ,
\end{equation}
where $\mT_2$ is the $2\times2$ coefficient tensor of the two-qubit state, with components
\begin{equation}
(\mT_2)_{jk} = t_{jk}\ .
\end{equation}
\\

\noindent
For a two-qubit mixed state, the convex-roof in Eq.~\eqref{eq:2-tangle-cr} can be computed as in Refs.~\cite{Hill:1997pfa, Wootters:1997id}: Let $\rho_{AB}$ be the density matrix of these two qubits, define the \textit{spin-flipped} density matrix as
\begin{equation}\label{eq:2-flip}
\tilde{\rho}_{AB} \equiv (\sigma_y\otimes\sigma_y)\rho_{AB}^*(\sigma_y\otimes\sigma_y)\ .
\end{equation}
The eigenvalues of $\rho_{AB}\tilde{\rho}_{AB}$ are real and nonnegative.
Let $\lambda_1, \lambda_2, \lambda_3, \lambda_4$ be the square roots of these eigenvalues in decreasing order.
Then the two-tangle~\eqref{eq:2-tangle-cr} can be obtained using the following formula~\cite{Hill:1997pfa, Wootters:1997id}
\begin{equation}\label{eq:2-tangle-cr-result}
E_{\tau_{A|B}}(\rho) = (\max\{0, \lambda_1 - \lambda_2 - \lambda_3 - \lambda_4\})^2\ .
\end{equation}
\subsection{Tripartite Entanglement Measures}
\label{subsec:tri}
\noindent
The three-tangle~\eqref{eq:CKW} is a genuine tripartite entanglement measure for three-qubit pure states~\cite{Coffman:1999jd}.
It can also be defined as the hyperdeterminant of the coefficient tensor.
For a general three-qubit pure state
\begin{equation}
|\Psi\rangle = \sum_{j,k, l=0}^1 t_{jkl}|jkl\rangle \ ,
\end{equation}
the three-tangle is given by~\cite{Coffman:1999jd}
\begin{equation}\label{eq:3-tangle}
\tau_{ABC} = 4|\Det\, \mT_3| = 4 \left| d_1 - 2 d_2 + 4 d_3 \right| \ ,
\end{equation}
with $\mT_3$ the $2\times2\times2$ coefficient tensor of the three-qubit pure state, $(\mT_3)_{jkl} = t_{jkl}$, and
\begin{equation}
\begin{split}
d_1 =&\, t_{000}^2 t_{111}^2 + t_{001}^2 t_{110}^2 + t_{010}^2 t_{101}^2 + t_{100}^2 t_{011}^2\ ,\\
d_2 =&\, t_{000} t_{111} t_{011} t_{100} + t_{000} t_{111} t_{101} t_{010} + t_{000} t_{111} t_{110} t_{001}\\
&+ t_{011} t_{100} t_{101} t_{010} + t_{011} t_{100} t_{110} t_{001} + t_{101} t_{010} t_{110} t_{001}\ ,\\
d_3 =&\, t_{000} t_{110} t_{101} t_{011} + t_{111} t_{001} t_{010} t_{100}\ .
\end{split}
\end{equation}
\\

\noindent
Another tripartite quantity is the third-order negativity $I_5$~\cite{Sudbery:2001jxc}.
In Ref.~\cite{Sudbery:2001jxc}, the author studied local invariants of three-qubit pure states, including two-tangles, the three-tangle~\eqref{eq:3-tangle}, and the third-order negativity $I_5$~\eqref{eq:I5-replica}.
Although $I_5$ is not an entanglement measure, combining it with the two- and three-tangles can form an entanglement measure $\phi_{ABC}$~\eqref{eq:phiABC}~\cite{Oreshkov:2006gqe}.
In Sec.~\ref{sec:phi-prop}, we show that $\phi_{ABC}$ is closely related to $I_5$ and maintains permutational symmetry for the three subsystems in a manifest form~\eqref{eq:phiABC-new}.
We also show that $\phi_{ABC}$ can serve as a necessary and sufficient criterion for the full separability of three-qubit mixed states.
Since $\phi_{ABC}$ is an entanglement measure, its role is not limited to providing a yes-or-no criterion for tripartite full separability. 
The value of $\phi_{ABC}$ is also meaningful, so it is natural to study its range and the states that saturate the corresponding bounds, as discussed in Sec.~\ref{sec:phi-prop} and Appendix~\ref{appx:phi-upper}.
\subsection{Quadripartite Entanglement Measures}
\noindent
Given the number of parameters, we still need 4 quadripartite entanglement measures.
We start with the four-qubit hyperdeterminant $\Delta$ and show how the remaining measures are hidden in it.
Consider a general four-qubit pure state
\begin{equation}\label{eq:4qubit-state}
\ket{\Psi} = \sum_{j,k,l,m = 0}^1 t_{jklm} \ket{jklm}\
\end{equation}
with the following notation
\begin{equation}\label{eq:binary-index}
t_{jklm} := a_r\ ,\quad r := 8j + 4k + 2l + m\ .
\end{equation}
We show that the degree-24 hyperdeterminant $\Delta$~\cite{Luque:2003dun} can be expressed in terms of objects with four-body permutation symmetry
\begin{equation}\label{eq:4-hyperdet}
\begin{split}
\Delta =&\, \Det\, \mT_4\\
=&\, -\frac{1}{108}\Big(4H^6\Pi + 6H^5\Sigma W - 3H^4\Sigma^2 - 48H^3\Pi W - 4H^3 W^3 + 48H^2\Pi\Sigma \\
&\qquad\quad\qquad - 60H^2\Sigma W^2 + 96H\Sigma^2 W + 64 \Pi^2 + 96\Pi W^2 - 32\Sigma^3 + 36W^4\Big)\ ,
\end{split}
\end{equation}
where $\Sigma$ and $\Pi$ can be expressed using degree-four invariants $L$, $M$, and $N$:
\begin{equation}\label{eq:Sigma-Pi}
\Sigma = L^2 + M^2 + N^2\ ,\quad \Pi = (L - M)(M - N)(N - L)\ .
\end{equation}
The degree-four invariants $L$, $M$, and $N$ are defined as
\begin{equation}\label{eq:Ldef-original}
L=
\det\!
\begin{pmatrix}
a_0 & a_4 & a_8 & a_{12}\\
a_1 & a_5 & a_9 & a_{13}\\
a_2 & a_6 & a_{10} & a_{14}\\
a_3 & a_7 & a_{11} & a_{15}
\end{pmatrix},\
M=
\det\!
\begin{pmatrix}
a_0 & a_8 & a_2 & a_{10}\\
a_1 & a_9 & a_3 & a_{11}\\
a_4 & a_{12} & a_{6} & a_{14}\\
a_5 & a_{13} & a_{7} & a_{15}
\end{pmatrix},\
N=
\det\!
\begin{pmatrix}
a_0 & a_1 & a_8 & a_{9}\\
a_2 & a_3 & a_{10} & a_{11}\\
a_4 & a_5 & a_{12} & a_{13}\\
a_6 & a_7 & a_{14} & a_{15}
\end{pmatrix},
\end{equation}
and they satisfy $L + M + N = 0$, thus only two of them are independent. The quantities $\Sigma$ and $\Pi$~\eqref{eq:Sigma-Pi} constructed from $L$, $M$, and $N$ are symmetric under permutations of $A$, $B$, $C$, and $D$.
Therefore, $\Sigma$ and $\Pi$ are quadripartite measures.
It has been shown in Ref.~\cite{Eltschka:2012voz} that
\begin{equation}\label{eq:LMN-det-rho}
|L|^2 = \det\rho_{AB}\ ,\quad |M|^2 = \det\rho_{AC}\ , \quad |N|^2 = \det\rho_{AD}\ .
\end{equation}
For a normalized $4\times 4$ matrix $\rho$, we have the following identity,
\begin{equation}
\det \rho = \frac{1}{24}\Big[1 - 6\Tr (\rho^2) + 3 \big(\!\Tr (\rho^2)\big)^2 + 8 \Tr (\rho^3) - 6 \Tr (\rho^4) \Big]\ ,
\end{equation}
thus we have
\begin{equation}
|L|^2 = \frac{1}{24}\Big[1 - 6\Tr(\rho_{AB}^2) + 3 \big(\!\Tr(\rho_{AB}^2)\big)^2 + 8 \Tr (\rho_{AB}^3) - 6 \Tr (\rho_{AB}^4) \Big]\ ,
\end{equation}
where $\Tr (\rho_{AB}^4)$ contains information about the fourth moment of the spectrum.
Therefore, $L$, $M$, and $N$ are independent of the bipartite and tripartite measures that are studied in Secs.~\ref{subsec:bi} and~\ref{subsec:tri}.
\\

\noindent
The quantity $W$ in Eq.~\eqref{eq:4-hyperdet} is a degree-6 quantity defined as~\cite{Tapia:2002du,Luque:2003dun,Miyake:2002vut}
\begin{equation}
W = D_{xy} + D_{xz} + D_{xw}\ ,
\end{equation}
with $D_{uv} = \det(B_{uv})$. The $3\times 3$ matrix $B_{xy}$ can be obtained by
\begin{equation}
\det\Big(\frac{\partial^2 A}{\partial z_i\partial w_j}\Big) = (x_0^2\quad x_0x_1\quad x_1^2) B_{xy}
\begin{pmatrix}
y_0^2\\
y_0y_1\\
y_1^2
\end{pmatrix}\ ,
\end{equation}
with $A$ a quadrilinear form
\begin{equation}
A(\mathbf{x}, \mathbf{y}, \mathbf{z}, \mathbf{w}) = \sum_{i,j,k,l=0}^1 t_{ijkl}x_i y_j z_k w_l\ .
\end{equation}
There are some relations between these quantities:
\begin{align}
&D_{xy} = D_{zw}\ , \quad D_{xz} = D_{yw}\ ,\quad D_{xw} = D_{yz}\ ,\\
&HL = D_{xz} - D_{xw}\ ,\quad HM = D_{xw} - D_{xy}\ ,\quad HN = D_{xy} - D_{xz}\ .
\end{align}
\\

\noindent
In Eq.~\eqref{eq:4-hyperdet}, $H$ is a degree-two invariant defined as~\cite{Luque:2003dun}
\bea
\label{eq:H-original}
H = a_0 a_{15} - a_1 a_{14} - a_2 a_{13} + a_3 a_{12} - a_4 a_{11} + a_5 a_{10} + a_6 a_9 - a_7 a_8\ .
\eea
We will show that the $H$ in Eq.~\eqref{eq:H-original} is precisely the quantity entering the four-tangle~\eqref{eq:4-tangle}, by expressing it in terms of the four-fold spin-flip operator $\Omega$~\cite{Uhlmann:2000ckp,Wong:2000cmz}:
\bea
\label{eq:H-spinflip-claim}
H=\frac{1}{2}\langle\Psi^*|\Omega|\Psi\rangle\ , \quad \Omega := \sigma_y \otimes \sigma_y \otimes \sigma_y \otimes \sigma_y\ .
\eea
For the mixed state $\rho_{ABCD}$, similar to Eq.~\eqref{eq:2-flip}, define
\bea
\tilde{\rho}_{ABCD}:=\sigma_y^{\otimes 4}\rho^*_{ABCD}\sigma_y^{\otimes 4}\ .
\eea
Then like Eq.~\eqref{eq:2-tangle-cr-result}, the four-tangle $4|H|^2$ can be computed from the square roots of $\rho_{ABCD}\tilde{\rho}_{ABCD}$ eigenvalues~\cite{Uhlmann:2000ckp,Wong:2000cmz}
\bea
4|H|^2 = C^2_{1234}(\rho_{ABCD}) = (\max\{0,\lambda_1 - \lambda_2 - \cdots - \lambda_{16}\})^2\ .
\eea
To verify Eq.~\eqref{eq:H-spinflip-claim}, recall that
\begin{equation}
\label{eq:sigmay-action}
\sigma_y\ket{0} = i\ket{1}\ ,\quad
\sigma_y\ket{1} = -i\ket{0}\ .
\end{equation}
Hence, one has
\begin{equation}
\label{eq:Omega-on-basis}
\Omega |jklm\rangle
=
i^4(-1)^{j+k+l+m}\,|\bar j\,\bar k\,\bar l\,\bar m\rangle
=
(-1)^{j+k+l+m}\,|\bar j\,\bar k\,\bar l\,\bar m\rangle\ ,
\end{equation}
where $\bar j:=1-j$.
Substituting Eq.~\eqref{eq:4qubit-state} into Eq.~\eqref{eq:H-spinflip-claim}, we obtain
\begin{align}
\langle\Psi^*|\Omega|\Psi\rangle
&=
\sum_{j,k,l,m,j',k',l',m'}
t_{jklm} t_{j'k'l'm'}
\langle jklm|\Omega|j'k'l'm'\rangle \nonumber\\
&=
\sum_{j,k,l,m=0}^1
(-1)^{j+k+l+m}
t_{jklm}t_{\bar j\,\bar k\,\bar l\,\bar m} \ .
\label{eq:psiOmegaPsi-sum}
\end{align}
Using Eq.~\eqref{eq:binary-index}, Eq.~\eqref{eq:psiOmegaPsi-sum} becomes
\begin{align}
\langle\Psi^*|\Omega|\Psi\rangle
=&
+a_0a_{15}-a_1a_{14}-a_2a_{13}+a_3a_{12}
-a_4a_{11}+a_5a_{10}+a_6a_9-a_7a_8 \nonumber\\
&
+
a_{15}a_0-a_{14}a_1-a_{13}a_2+a_{12}a_3
-a_{11}a_4+a_{10}a_5+a_9a_6-a_8a_7 \nonumber\\
=&\,
2\big(
a_0a_{15}-a_1a_{14}-a_2a_{13}+a_3a_{12}
-a_4a_{11}+a_5a_{10}+a_6a_9-a_7a_8
\big) \ .
\label{eq:psiOmegaPsi-expand}
\end{align}
Combining Eq.~\eqref{eq:psiOmegaPsi-expand} with Eq.~\eqref{eq:H-spinflip-claim}, we arrive at Eq.~\eqref{eq:H-original}.
\\

\noindent
Since the four-qubit hyperdeterminant $|\Delta| = |\Det \mT_4|$ has homogeneous degree $24$, Theorem~1 of Ref.~\cite{Eltschka:2012voz} rules it out as an entanglement measure.
By contrast, the degree-$4$ quantity $|\Delta|^{1/6}$ is one candidate that satisfies the general monotonicity criterion.
Similarly, since $\Sigma$ and $\Pi$ have degrees $8$ and $12$, respectively, the quantities $|\Sigma|^{1/2}$ and $|\Pi|^{1/3}$ are degree-4 entanglement measures.
Finally, the four-tangle $4|H|^2$ has degree $4$, thus it is also an entanglement measure\footnote{One can also choose $4|H|^2$, $|\Sigma|^{1/2}$, $|\Pi|^{1/3}$, and $|W|^{2/3}$ as the four independent quadripartite entanglement measures~\cite{Tapia:2002du,Luque:2003dun,Miyake:2002vut}, and here we choose the hyperdeterminant $\Delta$ instead of $W$. }.
\subsection{Examples}
\noindent
We compute all 18 entanglement measures for the following representative states:
\begin{itemize}
\item{
Generalized GHZ-State
\begin{equation}
\ket{\widetilde{\mathrm{GHZ}}^{(4)}} = a \ket{0000} + b \ket{1111}\ ;
\label{eq-gghz4Q}
\end{equation}}
\item{
W-State
\begin{equation}
\ket{\mathrm{W}} = \frac{1}{2} \big( \ket{0001} + \ket{0010} + \ket{0100} + \ket{1000} \big) \ ;
\label{eq-W4Q}
\end{equation}
}
\item{
Cluster-State
\begin{equation}
\ket{\mathrm{Cluster}} = \frac{1}{2} \big( \ket{0000} + \ket{0011} + \ket{1100} - \ket{1111} \big) \ ;
\label{eq-Cluster4Q}
\end{equation}
}
\item{
Dicke-State
\begin{equation}
\ket{\mathrm{Dicke}} = \frac{1}{\sqrt{6}}
\big(
\ket{0011} + \ket{0101} + \ket{0110}
+ \ket{1001} + \ket{1010} + \ket{1100}
\big) \ ;
\label{eq-Dicke42Q}
\end{equation}
}
\item{
Generalized Double Bell-State ($|a|^2 + |b|^2 = |c|^2 + |d|^2 = 1$)
\begin{equation}\label{eq:gdBell}
\ket{\widetilde{\mathrm{Bell}}}_{AB}\otimes\ket{\widetilde{\mathrm{Bell}}}_{CD} = \big( a\ket{00} + b\ket{11}\big)\otimes\big(c\ket{00} + d\ket{11} \big) \ .
\end{equation}
}
\end{itemize}
The results of the entanglement measures for these representative states are:
\begin{itemize}
\item{Generalized GHZ-state~\eqref{eq-gghz4Q},
\begin{equation}\label{eq:gghz-18}
\begin{split}
&\tau_{A|B} = \tau_{A|C} = \tau_{A|D} = \tau_{B|C} = \tau_{B|D} = \tau_{C|D} = 0\ ,\\
&\tau_{ABC} = \tau_{ABD} = \tau_{ACD} = \tau_{BCD} = 0\ ,\\
&\phi_{ABC} = \phi_{ABD} = \phi_{ACD} = \phi_{BCD} = 0\ ,\\
&H = ab\ ,\quad\Sigma = 0\ , \quad \Pi = 0\ ,\quad \Delta = 0\quad (W = 0)\ ;
\end{split}
\end{equation}
}
\item{W-state~\eqref{eq-W4Q},
\begin{equation}
\begin{split}
&\tau_{A|B} = \tau_{A|C} = \tau_{A|D} = \tau_{B|C} = \tau_{B|D} = \tau_{C|D} = 1/4\ ,\\
&\tau_{ABC} = \tau_{ABD} = \tau_{ACD} = \tau_{BCD} = 0\ ,\\
&\phi_{ABC} = \phi_{ABD} = \phi_{ACD} = \phi_{BCD} = 207/8\ ,\\
&H = 0\ ,\quad \Sigma = 0\ , \quad \Pi = 0\ ,\quad \Delta = 0\quad (W = 0)\ ;
\end{split}
\end{equation}
}
\item{Cluster-State~\eqref{eq-Cluster4Q},
\begin{equation}
\begin{split}
&\tau_{A|B} = \tau_{A|C} = \tau_{A|D} = \tau_{B|C} = \tau_{B|D} = \tau_{C|D} = 0\ ,\\
&\tau_{ABC} = \tau_{ABD} = \tau_{ACD} = \tau_{BCD} = 0\ ,\\
&\phi_{ABC} = \phi_{ABD} = \phi_{ACD} = \phi_{BCD} = 36\ ,\\
&H = 0\ , \quad \Sigma = \frac{1}{2^7}\ , \quad \Pi = \frac{1}{2^{11}}\ ,\quad \Delta = 0\quad (W = 0)\ ;
\end{split}
\end{equation}
}
\item{Dicke-State~\eqref{eq-Dicke42Q}:
\begin{equation}
\begin{split}
&\tau_{A|B} = \tau_{A|C} = \tau_{A|D} = \tau_{B|C} = \tau_{B|D} = \tau_{C|D} = 1/9\ ,\\
&\tau_{ABC} = \tau_{ABD} = \tau_{ACD} = \tau_{BCD} = 0\ ,\\
&\phi_{ABC} = \phi_{ABD} = \phi_{ACD} = \phi_{BCD} = 173/6\ ,\\
&H = 1/2\ ,\quad \Sigma = 0\ , \quad \Pi = 0\ ,\quad \Delta = 0\quad (W = 1/72)\ ;
\end{split}
\end{equation}
}
\item{Generalized Double Bell-State~\eqref{eq:gdBell}:
\begin{equation}
\begin{split}
&\tau_{A|B} = 4|ab|^2\ ,\quad \tau_{C|D} = 4|cd|^2\\
&\tau_{A|C} = \tau_{A|D} = \tau_{B|C} = \tau_{B|D} = 0\ ,\\
&\tau_{ABC} = \tau_{ABD} = \tau_{ACD} = \tau_{BCD} = 0\ ,\\
&\phi_{ABC} = \phi_{ABD} = 144 |ab|^2\ ,\quad \phi_{ACD} = \phi_{BCD} = 144|cd|^2 \ ,\\
&H = 2abcd\ ,\quad \Sigma = 2a^4 b^4 c^4 d^4\ , \quad \Pi = -2a^6 b^6 c^6 d^6\ ,\quad \Delta = 0\quad (W = 2a^3 b^3 c^3 d^3)\ .
\end{split}
\end{equation}
}
\end{itemize}
From the generalized GHZ-state, we find that only the four-tangle survives and quantifies quantum entanglement.
We discuss this observation further in Sec.~\ref{sec:H}.
The hyperdeterminant $\Delta$ always vanishes for these states.
Therefore, we are interested in further exploring the entanglement classes of four-qubit pure states~\cite{Verstraete:2002gqj}\footnote{Note that the states listed here have not been normalized. }:
\begin{equation}\label{eq-Gabcd1}
\begin{split}
\ket{G_{abcd}^1} = &\,\frac{a+d}{2}\big(\ket{0000}+\ket{1111}\big) + \frac{a-d}{2}\big(\ket{0011}+\ket{1100}\big)\\
&+ \frac{b+c}{2}\big(\ket{0101}+\ket{1010}\big) + \frac{b-c}{2}\big(\ket{0110}+\ket{1001}\big)\ ;
\end{split}
\end{equation}
\begin{equation}\label{eq-Gabc2}
\begin{split}
\ket{G_{abc}^2} =&\,\frac{a+b}{2}\big(\ket{0000} + \ket{1111}\big) + \frac{a-b}{2}\big(\ket{0011} + \ket{1100}\big)\\
&+ c\big(\ket{0101} + \ket{1010}\big) + \ket{0110}\ ;
\end{split}
\end{equation}
\begin{equation}\label{eq-Gab3}
\ket{G_{ab}^3} = a\big(\ket{0000} + \ket{1111}\big) + b\big(\ket{0101} + \ket{1010}\big) + \ket{0110} + \ket{0011}\ ;
\end{equation}
\begin{equation}\label{eq-Gab4}
\begin{split}
\ket{G_{ab}^4} =&\, a\big(\ket{0000} + \ket{1111}\big) + \frac{a + b}{2}\big(\ket{0101} + \ket{1010}\big) + \frac{a - b}{2}\big(\ket{0110} + \ket{1001}\big)\\
&+\frac{i}{\sqrt{2}}\big(\ket{0001} + \ket{0010} + \ket{0111} + \ket{1011}\big)\ ;
\end{split}
\end{equation}
\begin{equation}\label{eq-Ga5}
\ket{G_a^5} = a\big(\ket{0000} + \ket{0101} + \ket{1010} + \ket{1111}\big) + i\ket{0001} + \ket{0110} - i\ket{1011}\ ;
\end{equation}
\begin{equation}\label{eq-Ga6}
\ket{G_a^6} = a\big(\ket{0000} + \ket{1111}\big) + \ket{0011} + \ket{0101} + \ket{0110}\ ;
\end{equation}
\begin{equation}\label{eq-G7}
\ket{G^7} = \ket{0000} + \ket{0101} + \ket{1000} + \ket{1110}\ ;
\end{equation}
\begin{equation}\label{eq-G8}
\ket{G^8} = \ket{0000} + \ket{1011} + \ket{1101} + \ket{1110}\ ;
\end{equation}
\begin{equation}\label{eq-G9}
\ket{G^9} = \ket{0000} + \ket{0111}\ .
\end{equation}

\noindent
The quadripartite entanglement measures are listed in Table~\ref{tab-4-EM}, and $\Delta$ is non-zero only for $G_{abcd}^1$~\cite{Luque:2003dun}.
Hence, $\Delta$ characterizes a very specific type of quadripartite entanglement.
\begin{table}[htpb]
\noindent
\hspace*{\dimexpr(\textwidth-\paperwidth)/2\relax}%
\makebox[\paperwidth][c]{
\begin{tabular}{lccccc}
\hline
State & $H$ & $L$ & $M$ & $D_{xw}$ & $\Delta$\\
\hline
$G_{abcd}^1$~\eqref{eq-Gabcd1} & $\frac{a^2 + b^2 +c^2 + d^2}{2}$ & $abcd$ & $[(\frac{c - d}{2})^2 \!-\! (\frac{a - b}{2})^2][(\frac{a + b}{2})^2 \!-\! (\frac{c + d}{2})^2]$ & $\frac{(ad-bc)(bd-ac)(ab-cd)}{4}$ & $\frac{V(a^2,b^2,c^2,d^2)}{256}$\\
$G_{abc}^2$~\eqref{eq-Gabc2} & $\frac{a^2 + b^2 + 2c^2}{2}$ & $abc^2$ & $[c^2-(\frac{a+b}{2})^2](\frac{a - b}{2})^2$ & $\frac{c^2(a-b)^2(c^2 - ab)}{4}$ & $0$\\
$G_{ab}^3$~\eqref{eq-Gab3} & $a^2 + b^2$ & $a^2b^2$ & $0$ & $0$ & $0$\\
$G_{ab}^4$~\eqref{eq-Gab4} & $\frac{3a^2 + b^2}{2}$ & $a^3b$ & $[a^2-(\frac{a + b}{2})^2](\frac{a - b}{2})^2$ & $\frac{a^3(a - b)^3}{4}$ & $0$\\
$G_a^5$~\eqref{eq-Ga5} & $2a^2$ & $a^4$ & $0$ & $0$ & $0$\\
$G_a^6$~\eqref{eq-Ga6} & $a^2$ & $0$ & $0$ & $0$ & $0$\\
$G^7$~\eqref{eq-G7} & $0$ & $0$ & $0$ & $0$ & $0$\\
$G^8$~\eqref{eq-G8} & $0$ & $0$ & $0$ & $0$ & $0$\\
$G^9$~\eqref{eq-G9} & $0$ & $0$ & $0$ & $0$ & $0$\\
\hline
\end{tabular}}
\caption{Quadripartite entanglement measures for all different entangled classes~\cite{Luque:2003dun}.
The Vandermonde determinant is $V(a^2,b^2,c^2,d^2) = (a^2 - b^2)^2 (a^2 - c^2)^2 (a^2 - d^2)^2 (b^2 - c^2)^2 (b^2 - d^2)^2 (c^2 - d^2)^2$. }
\label{tab-4-EM}
\end{table}
\subsection{Independent Entanglement Measure}
\label{sec:H}
\noindent
The next question is whether one of these measures can exist independently, namely, whether we can keep a single entanglement measure nonvanishing while turning off all the others.
Among the eighteen measures, the four quantities $\phi_{ABC}$, $\phi_{ABD}$, $\phi_{ACD}$, and $\phi_{BCD}$ play a distinguished role.
As shown in Sec.~\ref{sec:phi-prop}, each of them provides a necessary and sufficient criterion for the full separability of the corresponding three-qubit mixed state.
Therefore, to turn off all four $\phi$'s, one is forced to require all four three-qubit reduced density matrices to be fully separable.
Ref.~\cite{Thapliyal:1998nw} shows that: \textit{For a $\mathtt{q}$-qudit pure state, suppose that, for every qudit $X$, the $(\mathtt{q} - 1)$-qudit reduced density matrix $\rho_{\bar{X}} = \Tr_X |\Psi\rangle\langle\Psi|$ is fully separable, then $|\Psi\rangle$ is LU-equivalent to a $\mathtt{q}$-qudit generalized GHZ-state
\begin{equation}\label{eq:qgghz}
\ket{\widetilde{\mathrm{GHZ}}^{(\mathtt{q})}} = \sum_{j} \lambda_j|j\rangle^{\otimes \mathtt{q}}\ ,
\end{equation}
and conversely, any state of the form~\eqref{eq:qgghz} satisfies the above property.\footnote{We provide an alternative proof of this proposition in Sec.~\ref{sec:qqudit-ghz}. }}
\\

\noindent
For $\mathtt{q} = 4$, this implies that any four-qubit pure state with
\begin{equation}
\phi_{ABC} = \phi_{ABD} = \phi_{ACD} = \phi_{BCD} = 0
\end{equation}
is LU-equivalent to a four-qubit generalized GHZ-type state
\begin{equation}
\ket{\widetilde{\mathrm{GHZ}}^{(4)}} = a \ket{0000} + b \ket{1111}\ ,
\end{equation}
whose 18 entanglement measures are given in Eq.~\eqref{eq:gghz-18}. 
Except for $H = ab$, all the other 17 measures vanish.
This observation already shows that among the eighteen entanglement measures, the four-tangle can survive independently, and that this occurs uniquely for the generalized GHZ-state.
We now show that this is in fact the \textit{only} possibility, by showing that $\phi_{ABC}$ cannot exist independently:
\begin{proposition}\label{prop:5conditions}
There is no four-qubit pure state $|\Psi\rangle \in (\mathbb{C}^2)^{\otimes 4}$ such that
\begin{enumerate}[label=(\roman*)]
\item $\rho_{BCD} = \Tr_A |\Psi\rangle\langle\Psi|$ is fully separable,
\item $\rho_{ACD} = \Tr_B |\Psi\rangle\langle\Psi|$ is fully separable,
\item $\rho_{ABD} = \Tr_C |\Psi\rangle\langle\Psi|$ is fully separable,
\item $\rho_{ABC} = \Tr_D |\Psi\rangle\langle\Psi|$ is not fully separable,
\item $\tau_{ABC} = \tau_{ABD} = \tau_{ACD} = \tau_{BCD} = 0$.
\end{enumerate}
\end{proposition}
\begin{proof}
Assume, for contradiction, that such a state $|\Psi\rangle$ exists.
By assumption (i), $\rho_{BCD}$ is a fully separable three-qubit state with rank at most  2.
Since $\rho_{ABC}$ is not fully separable, we further require that $\rank(\rho_{BCD}) = 2$, otherwise $\ket{\Psi}$ is a fully product state.
Therefore, one can write
\begin{equation}\label{eq:Psi-two-term}
|\Psi\rangle
=
|a_0\rangle |b_0\rangle |c_0\rangle |d_0\rangle
+
|a_1\rangle |b_1\rangle |c_1\rangle |d_1\rangle\ .
\end{equation}

\bigskip

\noindent
We now prove a simple lemma.
\begin{lemma}\label{lemma:sep}
Let $|x_0\rangle , |x_1\rangle$ be linearly independent vectors in $\mathcal{H}_X$, and $|y_0\rangle , |y_1\rangle$ be linearly independent vectors in $\mathcal{H}_Y$. Consider a positive semidefinite operator $\sigma$
\begin{equation}\label{eq:2term-lemma-state}
\sigma
=
\mu_0|x_0 y_0\rangle\langle x_0 y_0|
+
\mu_1|x_1 y_1\rangle\langle x_1 y_1|
+
s\, |x_0 y_0\rangle\langle x_1 y_1|
+
s^*\, |x_1 y_1\rangle\langle x_0 y_0| \ ,
\end{equation}
with $\mu_0,\mu_1\ge 0$.
Then $\sigma$ is separable across $X|Y$ if and only if $s=0$.
\end{lemma}
\begin{proof}
By invertible local transformations, $\sigma$ is equivalent to
\begin{equation}
\mu_0|00\rangle\langle 00|
+
\mu_1|11\rangle\langle 11|
+
s' |00\rangle\langle 11|
+
s'^* |11\rangle\langle 00| \ ,
\end{equation}
where $s'=0$ if and only if $s=0$. Its partial transpose has eigenvalues $\{\mu_0, \mu_1, |s'|, -|s'|\}$.
Hence, it is PPT if and only if $s'=0$. By the $2\times 2$ PPT criterion~\cite{Peres:1996dw,Horodecki:1996nc}, this is equivalent to separability.
\end{proof}
\bigskip
\noindent
Tracing out subsystem $B$ in Eq.~\eqref{eq:Psi-two-term} and applying Lemma~\ref{lemma:sep} with assumption (ii), one can obtain $\langle b_1|b_0\rangle=0$.
Similarly, tracing out subsystem $C$ or subsystem $A$ gives $\langle c_1|c_0\rangle=0$, $\langle a_1|a_0\rangle=0$, respectively.
Therefore, up to local unitaries,
\begin{equation}\label{eq:GHZ-type}
|\Psi\rangle
=
|000\rangle_{ABC}|d_0\rangle
+
|111\rangle_{ABC}|d_1\rangle\ .
\end{equation}

\noindent
Tracing out subsystem $D$ in Eq.~\eqref{eq:GHZ-type}, we find
\begin{equation}\label{eq:rhoABC-final}
\rho_{ABC}
=
\langle d_0 | d_0\rangle |000\rangle\langle 000|
+
\langle d_1 | d_1\rangle |111\rangle\langle 111|
+
\langle d_1 | d_0\rangle |000\rangle\langle 111|
+
\langle d_0 | d_1\rangle |111\rangle\langle 000| \ .
\end{equation}
If $\langle d_0 | d_1\rangle=0$, then $\rho_{ABC}$ in Eq.~\eqref{eq:rhoABC-final} is fully separable, contradicting assumption (iv).
Hence, we must have $\langle d_0 | d_1\rangle\neq 0$, then $\rho_{ABC}$ is a rank-$2$ GHZ-type state supported on
\begin{equation}
\mathrm{span}\{|000\rangle,|111\rangle\}\ ,
\end{equation}
with nonzero coherence between $|000\rangle$ and $|111\rangle$. In this sense, $\phi_{ABC}$ is not zero. Such a state also has a nonvanishing three-tangle.
This contradicts assumption (v).
Therefore, assumptions (i)--(v) cannot be satisfied simultaneously.
\end{proof}

\noindent
Our result shows that the four-tangle $4|H|^2$~\cite{Uhlmann:2000ckp,Wong:2000cmz} is the only entanglement measure among the eighteen that can exist independently.
To quantify quantum entanglement, a state should involve only one entanglement measure.
Otherwise, the quantification remains ambiguous under the current understanding.
Hence, our results show that the generalized GHZ-state is the only state for which we can quantify quantum entanglement unambiguously using the four-tangle.
\section{Generalization to Qudit Systems}
\label{sec:multipartite}
\noindent
So far, the discussion has focused on tripartite full separability and its special role in applications to four-qubit pure states.
However, the underlying logic is not restricted to these low-party and few-qubit examples.
The third-order negativity criterion can be extended to arbitrary finite-dimensional multipartite systems.
It is therefore natural to ask the extent to which the separability criteria above admit a general multipartite formulation.
In this section, we show that the answer is affirmative.
\\

\noindent
We first generalize the tripartite construction and obtain necessary and sufficient criteria for full separability in general finite-dimensional multipartite systems.
We then show that the same framework also leads to a multipartite rigidity statement: for a $\mathtt{q}$-partite pure state, if all its $(\mathtt{q}-1)$-partite reduced density matrices are fully separable, the state must be LU-equivalent to a generalized GHZ-state.
This result is particularly interesting for extensions to field theory.
Because the local QFT has no separable state for the spatial partition~\cite{Witten:2018zxz}, it cannot have a generalized GHZ-state.
\subsection{Replica Criterion for Fully Separable Qudit States}
\label{sec:replica-qudit-product}
\noindent
Proposition~\ref{prop:I5-product} shows that the condition
\begin{equation}
I_5 = 1
\end{equation}
provides a criterion for tripartite fully product states.
Since $I_5$ admits a replica construction~\eqref{eq:I5-replica}, it is natural to ask whether, for \textit{multipartite} pure states, one can construct analogous replica quantities that serve as criteria for \textit{multipartite} fully product states.
\\

\noindent
We consider $\mathcal{Z}_{\Omega_{1,\cdots,\mathtt{q}}}^{(\mathtt{q})}$~\eqref{eq:general-partite-function}, called a multi-invariant~\cite{Gadde:2022cqi,Gadde:2023zzj}, to generalize Proposition~\ref{prop:I5-product} to the $\mathtt{q}$-partite setting: $|\mathcal{Z}_{\Omega_{1,\cdots,\mathtt{q}}}^{(\mathtt{q})}| = 1$ can serve as a criterion for $\mathtt{q}$-partite fully product states (Proposition~\ref{prop:general-replica-product}).
Let
\begin{equation}\label{eq:qudit-state1}
|\Psi\rangle
=
\sum_{i_1=0}^{d_1-1}\cdots \sum_{i_{\mathtt{q}}=0}^{d_{\mathtt{q}}-1}
t_{i_1\cdots i_{\mathtt{q}}}
|i_1\cdots i_{\mathtt{q}}\rangle
\in
\mathbb{C}^{d_1}\otimes\cdots\otimes\mathbb{C}^{d_{\mathtt{q}}}
\end{equation}
be a normalized $\mathtt{q}$-qudit pure state.
Let $N$ be the number of replicas, and let
\begin{equation}\label{eq:general-Omegas}
\Omega_r \in S_N\ \ \text{for}\ \ r=1,\dots,\mathtt{q}\ ,
\end{equation}
with
\begin{equation}\label{eq:general-Omega-condition}
\Omega_r \neq \Omega_s\ \ \text{for all}\ \ 1\le r<s\le \mathtt{q} \ .
\end{equation}
Define
\begin{equation}\label{eq:general-partite-function}
\mathcal{Z}_{\Omega_{1,\cdots,\mathtt{q}}}^{(\mathtt{q})}(\Psi)
\equiv
\langle\Psi|^{\otimes N}
\Omega_1\Omega_2\cdots\Omega_{\mathtt{q}}
|\Psi\rangle^{\otimes N}\ ,
\end{equation}
where $\Omega_r$ acts on the $r$-th subsystem across the $N$ replicas. Note that the quantity in Eq.~\eqref{eq:general-partite-function} is invariant under $\Omega_{1,\cdots,\mathtt{q}}\to \Omega'\Omega_{1,\cdots,\mathtt{q}}$, thus one can always set $\Omega_1 = \mathrm{id}$ by taking $\Omega' = (\Omega_1)^{-1}$.
\begin{proposition}\label{prop:general-replica-product}
For a normalized pure state $|\Psi\rangle \in \mathbb{C}^{d_1}\otimes\cdots\otimes\mathbb{C}^{d_{\mathtt{q}}}$,
\begin{equation}\label{eq:general-product-criterion}
\big|\mathcal{Z}_{\Omega_{1,\cdots,\mathtt{q}}}^{(\mathtt{q})}(\Psi)\big|=1
\quad\Longleftrightarrow\quad
|\Psi\rangle \ \text{is fully product} \ .
\end{equation}
\end{proposition}

\begin{proof}
The direction ``fully product $\Rightarrow \big|\mathcal{Z}_{\Omega_{1,\cdots,\mathtt{q}}}^{(\mathtt{q})}(\Psi)\big|=1$'' is immediate. We now prove the converse. Let
\begin{equation}\label{eq:Omega-total-general}
\Omega := \Omega_1\Omega_2\cdots\Omega_{\mathtt{q}} \ .
\end{equation}
Since $\Omega$ is unitary, one immediately has
\begin{equation}\label{eq:general-upper-bound}
\big|\mathcal{Z}_{\Omega_{1,\cdots,\mathtt{q}}}^{(\mathtt{q})}(\Psi)\big|\le 1 \ .
\end{equation}
\\

\noindent
If $|\mathcal{Z}_{\Omega_{1,\cdots,\mathtt{q}}}^{(\mathtt{q})}(\Psi)|=1$, then $|\Psi\rangle^{\otimes N}$ is an eigenvector of $\Omega$ with eigenvalue $e^{i\theta}$, namely
\begin{equation}\label{eq:Omega-fixed-vector-general}
\Omega |\Psi\rangle^{\otimes N}
=
e^{i\theta}|\Psi\rangle^{\otimes N}\ ,\quad \theta\in[0,2\pi)\ .
\end{equation}
We first rewrite Eq.~\eqref{eq:Omega-fixed-vector-general} as an identity for the coefficients $t_{i_1\cdots i_{\mathtt{q}}}$. For each replica $x\in \{1,\dots,N\}$, we choose local indices
\begin{equation}\label{eq:index-assignment-general-2}
i_r(x)\in\{0,\dots,d_r-1\} \ , \quad r=1,\dots,\mathtt{q} \ .
\end{equation}
Comparing the coefficient of the corresponding computational basis vector on both sides of Eq.~\eqref{eq:Omega-fixed-vector-general}, we obtain
\begin{equation}\label{eq:master-identity1}
\prod_{x=1}^{N}
t_{\,i_1(\Omega_1(x)),\,i_2(\Omega_2(x)),\,\dots,\,i_{\mathtt{q}}(\Omega_{\mathtt{q}}(x))}
=
e^{i\theta}\prod_{x=1}^{N}
t_{\,i_1(x),\,i_2(x),\,\cdots,\,i_{\mathtt{q}}(x)} \ .
\end{equation}
We next show that $\theta$ must be zero. Since $|\Psi\rangle$ is normalized, there exists at least one choice of indices
\begin{equation}\label{eq:assignment-all}
(i_1,\dots,i_{\mathtt{q}}) = (\mu_1,\dots,\mu_{\mathtt{q}})
\end{equation}
such that $t_{\mu_1\cdots \mu_{\mathtt{q}}}\neq 0$.
Now we apply the assignment~\eqref{eq:assignment-all} for all $x=1,\dots,N$, then Eq.~\eqref{eq:master-identity1} gives
\begin{equation}
\bigl(t_{\mu_1\cdots \mu_{\mathtt{q}}}\bigr)^N
=
e^{i\theta}
\bigl(t_{\mu_1\cdots \mu_{\mathtt{q}}}\bigr)^N \ .
\end{equation}
Since $t_{\mu_1\cdots \mu_{\mathtt{q}}}\neq 0$, this implies $e^{i\theta}=1$ and $\theta = 0$.
\\

\noindent
Since $\Omega$ commutes with $U^{\otimes N}$ for every local unitary $U=U_1\otimes\cdots\otimes U_\mathtt{q}$, Eq.~\eqref{eq:Omega-fixed-vector-general} is preserved
under local changes of basis. We may therefore choose a generic
local product basis in which all coefficients
$t_{i_1\cdots i_{\mathtt{q}}}$ are nonzero, and retain the same notation
for these coefficients.
\\

\noindent
Let $T$ be the coefficient tensor of the state~\eqref{eq:qudit-state1},
with $(T)_{i_1\cdots i_{\mathtt{q}}}=t_{i_1\cdots i_{\mathtt{q}}}$.
We finally show that every two-index slice of $T$ has rank one.
Consider a slice involving the $r$-th and $s$-th indices, where $1\le r<s\le \mathtt{q}$. Since $\Omega_r\neq \Omega_s$, choose $x_*$ such that
\begin{equation}
\Omega_r^{-1}(x_*)\neq \Omega_s^{-1}(x_*) \ .
\end{equation}
Fix all indices except $(i_r,i_s)$, and choose two values
\begin{equation}
\alpha_r\neq \beta_r \in \{0,\dots,d_r-1\} \ ,
\quad
\alpha_s\neq \beta_s \in \{0,\dots,d_s-1\} \ .
\end{equation}
Let
\begin{equation}\label{eq:ABCD-symmetric}
\begin{split}
&A := t_{\mu_1\dots\mu_{r-1}\alpha_r\mu_{r+1}\cdots\mu_{s-1}\alpha_s\mu_{s+1}\cdots\mu_{\mathtt{q}}}\ , \quad
B := t_{\mu_1\dots\mu_{r-1}\alpha_r\mu_{r+1}\cdots\mu_{s-1}\beta_s\mu_{s+1}\cdots\mu_{\mathtt{q}}}\ , \\
&C := t_{\mu_1\dots\mu_{r-1}\beta_r\mu_{r+1}\cdots\mu_{s-1}\alpha_s\mu_{s+1}\cdots\mu_{\mathtt{q}}}\ , \quad
D := t_{\mu_1\dots\mu_{r-1}\beta_r\mu_{r+1}\cdots\mu_{s-1}\beta_s\mu_{s+1}\cdots\mu_{\mathtt{q}}}\ ,
\end{split}
\end{equation}
where $\mu_{\dots}$ denotes the fixed values of the remaining $\mathtt q-2$ indices.
We claim that
\begin{equation}\label{eq:minor-symmetric}
BC = AD \ .
\end{equation}
Since all coefficients are nonzero, we have $D\neq0$.
Choose a copy labeled by $x_*\in \{1,\dots,N\}$. 
Apply the following assignment
\begin{equation*}
\begin{aligned}
&\text{at } x_* : \ (i_{1,\dots,r-1},i_r,i_{r+1,\dots,s-1},i_s,i_{s+1,\dots,\mathtt{q}}) = (\mu_{1,\dots,r-1},\alpha_r,\mu_{r+1,\cdots,s-1},\alpha_s,\mu_{s+1,\cdots,\mathtt{q}})\ , \\
&\text{at } x\neq x_* : \ (i_{1,\dots,r-1},i_r,i_{r+1,\dots,s-1},i_s,i_{s+1,\dots,\mathtt{q}}) = (\mu_{1,\dots,r-1},\beta_r,\mu_{r+1,\cdots,s-1},\beta_s,\mu_{s+1,\cdots,\mathtt{q}})
\end{aligned}
\end{equation*}
to Eq.~\eqref{eq:master-identity1}, then the LHS of Eq.~\eqref{eq:master-identity1} is $B C D^{N-2}$, while the RHS is $A D^{N-1}$.\footnote{Recall that we have already proved that $\theta = 0$ in Eq.~\eqref{eq:master-identity1}. 
If we set $\alpha_r = \beta_r$, we get $AB = BA$ instead of Eq.~\eqref{eq:minor-symmetric}, as we require.}
Since $D\neq 0$, we obtain Eq.~\eqref{eq:minor-symmetric}. 
Therefore, the determinant of every $2\times2$ submatrix of every two-index slice of $T$ vanishes.
Since all coefficients are nonzero, Eq.~\eqref{eq:minor-symmetric} gives
\[
    \frac{t_{\cdots i_r\cdots i_s\cdots}}
         {t_{\cdots \mu_r\cdots i_s\cdots}}
    =
    \frac{t_{\cdots i_r\cdots \mu_s\cdots}}
         {t_{\cdots \mu_r\cdots \mu_s\cdots}}\ ,
    \quad r\neq s\ ,
\]
where all undisplayed indices are held fixed.
Thus, the ratio obtained by changing the $r$-th index from $\mu_r$ to $i_r$ is unchanged when any other index is varied.
Fix a reference multi-index
$\boldsymbol{\mu}=(\mu_1,\ldots,\mu_\mathtt{q})$.
By successively setting all indices other than the $r$-th
to their reference values, we find
\[
    \frac{t_{i_1\cdots i_r\cdots i_\mathtt{q}}}
         {t_{i_1\cdots \mu_r\cdots i_\mathtt{q}}}
    =
    \frac{
        t_{\mu_1\cdots\mu_{r-1}\,i_r\,\mu_{r+1}
        \cdots\mu_\mathtt{q}}
    }{
        t_{\mu_1\cdots\mu_\mathtt{q}}
    }
    \equiv R_r(i_r).
\]
Starting from $t_{\mu_1\cdots\mu_\mathtt{q}}$ and changing
the indices one at a time, each change therefore contributes
a factor $R_r(i_r)$, independently of the indices already changed.
Consequently,
\[
    t_{i_1\cdots i_\mathtt{q}}
    =
    t_{\mu_1\cdots\mu_\mathtt{q}}
    \prod_{r=1}^{\mathtt{q}} R_r(i_r)
    =
    t_{\mu_1\cdots\mu_\mathtt{q}}^{\,1-\mathtt{q}}
    \prod_{r=1}^{\mathtt{q}}
    t_{\mu_1\cdots\mu_{r-1}\,i_r\,\mu_{r+1}
    \cdots\mu_\mathtt{q}}\ .
\]
Each factor in the product depends on only one index.
Absorbing the overall constant into one of these factors,
we conclude that $T$ has rank one, i.e.,
\begin{equation}\label{eq:T-rank-one}
t_{i_1\cdots i_{\mathtt q}}
=
u^{(1)}_{i_1}\cdots u^{(\mathtt q)}_{i_{\mathtt q}}
\end{equation}
for some vectors $u^{(r)}\in\mathbb{C}^{d_r}$.
Therefore, we get
\begin{equation}\label{eq:psi-product-final}
|\Psi\rangle
=
\Biggl( \sum_{i_1=0}^{d_1-1} u^{(1)}_{i_1}|i_1\rangle \Biggr)
\otimes \cdots \otimes
\Biggl( \sum_{i_{\mathtt q}=0}^{d_{\mathtt q}-1} u^{(\mathtt q)}_{i_{\mathtt q}}|i_{\mathtt q}\rangle \Biggr) \ ,
\end{equation}
so $|\Psi\rangle$ is fully product.
\end{proof}

\noindent
The key point of this construction is that the extremal condition $|\mathcal{Z}|=1$ turns a replica invariant into a rigid fixed-vector constraint on the replicated state, i.e., Eq.~\eqref{eq:Omega-fixed-vector-general} with $\theta = 0$.
This yields polynomial relations among the coefficients of $|\Psi\rangle$. In the local basis chosen above, where all coefficients are nonzero, these relations imply the factorization in Eq.~\eqref{eq:T-rank-one}. Hence, $|\Psi\rangle$ is fully product.
Since the symmetric group $S_N$ has $N!$ elements, in principle one needs at least the smallest integer $N$ such that $N! \ge\mathtt{q}$ replicas to determine whether a $\mathtt{q}$-partite state is fully product.
\\

\noindent
Several important examples are listed below:
\begin{itemize}[leftmargin=\parindent,labelsep=0.5em,itemsep=0.4em]
\item Bipartite $n$-th R\'enyi entropy: $N = n$, $\Omega_1 = \mathrm{id}$, and $\Omega_2 = (12\cdots n)$.
The corresponding replica partition function $\mathcal{Z}_{n}^{(2)}$ is a real number.
\item Third-order negativity~\eqref{eq:I5-replica}: $N = 3$, $\Omega_1 = \mathrm{id}$, $\Omega_2 = (123)$, and $\Omega_3 = (132)$.
The corresponding replica partition function $\mathcal{Z}_{\Omega_{1,2,3}}^{(3)}$ is a real number.
\item $\mathtt{q}$-partite $n$-th R\'enyi multi-entropy~\cite{Gadde:2022cqi}: there are $N = n^{\mathtt{q} - 1}$ replicas, and we label these replicas by vectors
\bea
x=(x_2,\dots,x_{\mathtt{q}})\in \mathbb{Z}_{n}^{\mathtt{q} - 1}\ .
\eea
Let $\Omega_1 = \mathrm{id}$, and for each $r=2,\dots,\mathtt{q}$, let $\Omega_r$ be the permutation operator that shifts the $r$-th hypercube coordinate,
\begin{equation}\label{eq:Omega-r}
\Omega_r:\ x \mapsto x+e_r \ ,
\end{equation}
where $e_r$ is the unit vector in the $r$-th direction.
The corresponding replica partition function
\begin{equation}\label{eq:partite-function}
\mathcal{Z}_{n}^{(\mathtt{q})}(\Psi)
\equiv
\langle\Psi|^{\otimes N}
\Omega_1\Omega_2\cdots\Omega_{\mathtt{q}}
|\Psi\rangle^{\otimes N}
\end{equation}
is a real number. See Fig.~\ref{fig:Z-tripartite} for the graphical notation of the tripartite wavefunction, reduced density matrix, $\mathcal{Z}_2^{(3)}$ and $\mathcal{Z}_3^{(3)}$ from the contractions of reduced density matrices. See Fig.~\ref{fig:Z-quadripartite} for the $\mathtt{q} = 4$, $n = 2$ case. 
\end{itemize}

\noindent
The three examples above share two features that make them especially natural: their corresponding replica partition functions are directly real-valued, and they are invariant under arbitrary permutations of the parties.\footnote{Invariance under permutations of the parties means that
\[
\langle\Psi|^{\otimes N}
\Omega_1\Omega_2\cdots\Omega_\mathtt{q}
|\Psi\rangle^{\otimes N}
=
\langle\Psi|^{\otimes N}
\Omega_{\pi(1)}\Omega_{\pi(2)}\cdots\Omega_{\pi(\mathtt{q})}
|\Psi\rangle^{\otimes N}
\]
for every $\pi\in S_\mathtt{q}$, where $\Omega_{\pi(r)}$ on the
right-hand side acts on the $r$-th subsystem.
In general, the multi-invariants in Eq.~\eqref{eq:general-partite-function} need not possess this symmetry.}
As in Sec.~\ref{sec:mixed-3}, Proposition~\ref{prop:general-replica-product} also generalizes to mixed states via the convex-roof construction.

\begin{proposition}\label{prop:qudit-sep-criterion1}
For every normalized $\mathtt{q}$-qudit mixed state, let
\begin{equation}\label{eq:E1-Z}
E_{1 - |\mathcal{Z}_{\Omega_{1,\cdots,\mathtt{q}}}^{(\mathtt{q})}|}(\rho) \equiv \inf_{\{p_\ell,\ket{\psi_\ell}\}\in \mathbb{D(\rho)}}\sum_{\ell} p_{\ell}\Big(1 - \big|\mathcal{Z}_{\Omega_{1,\cdots,\mathtt{q}}}^{(\mathtt{q})}(\psi_{\ell})\big|\Big)\ ,
\end{equation}
then,
\begin{equation}\label{eq:qudit-sep-criterion}
E_{1 - |\mathcal{Z}_{\Omega_{1,\cdots,\mathtt{q}}}^{(\mathtt{q})}|}(\rho) = 0 \iff \rho\ \text{is fully separable}\ .
\end{equation}
\end{proposition}
\begin{proof}
This follows directly from Eq.~\eqref{eq:general-upper-bound} and Proposition~\ref{prop:general-replica-product}.
\end{proof}

\begin{figure}[tbp]
\centering
\hspace*{\dimexpr(\textwidth-\paperwidth)/2\relax}%
\makebox[\paperwidth][c]{
\includegraphics[scale = 1]{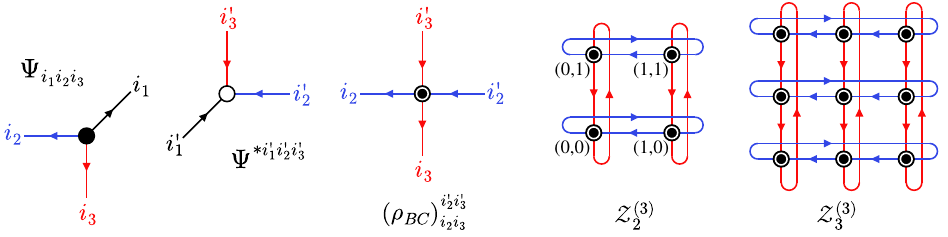}}
\caption{Graphical notation for the tripartite wavefunction $\Psi$, its conjugate $\Psi^*$, the reduced density matrix $\rho_{BC}$, and the construction of~$\mathcal{Z}^{(\mathtt{q = 3})}_{n = 2}$ and $\mathcal{Z}^{(\mathtt{q = 3})}_{n = 3}$. } 
\label{fig:Z-tripartite}
\end{figure}

\begin{figure}[tbp]
\centering
\includegraphics[scale = 1]{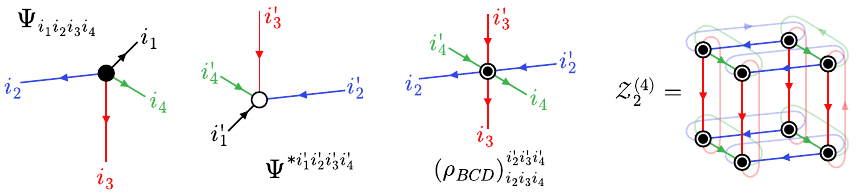}
\caption{Graphical notation for the quadripartite wavefunction $\Psi$, its conjugate $\Psi^*$, the reduced density matrix $\rho_{BCD}$, and the construction of~$\mathcal{Z}^{(\mathtt{q = 4})}_{n = 2}$.} 
\label{fig:Z-quadripartite}
\end{figure}
\subsection{Subsystem Separability and Generalized GHZ-States}
\label{sec:qqudit-ghz}
\noindent
Here we provide an alternative proof of the proposition in Ref.~\cite{Thapliyal:1998nw}.
Namely, for a $\mathtt{q}$-partite pure state, requiring all $(\mathtt{q}-1)$-partite reduced density matrices to be fully separable is an extremely restrictive condition that forces the state to be LU-equivalent to a generalized GHZ-state.
We formulate this precisely in the following proposition.
\begin{proposition}\label{prop:qqudit-ghz}
Let
\begin{equation}
|\Psi\rangle \in \mathcal{H}_1 \otimes \mathcal{H}_2 \otimes \cdots \otimes \mathcal{H}_{\mathtt{q}}
\end{equation}
be a $\mathtt{q}$-qudit pure state, where all Hilbert spaces are finite-dimensional.
Suppose that for every qudit $X$, the $(\mathtt{q} - 1)$-qudit reduced density matrix
\begin{equation}
\rho_{\bar{X}} = \Tr_X |\Psi\rangle\langle\Psi|
\end{equation}
is fully separable.
Then $|\Psi\rangle$ is LU-equivalent to a rank-$R$ $\mathtt{q}$-qudit generalized GHZ-state
\begin{equation}\label{eq:qgghzr}
\ket{\widetilde{\mathrm{GHZ}}_R^{(\mathtt{q})}} = \sum_{j = 0}^{R - 1} \lambda_j|j\rangle^{\otimes \mathtt{q}}\ ,
\end{equation}
where $\sum_j |\lambda_j|^2=1$ and
\begin{equation}
R \le \min\{\dim\mathcal{H}_1,\dim\mathcal{H}_2,\dots,\dim\mathcal{H}_{\mathtt{q}}\} \ .
\end{equation}
Conversely, any state of the form~\eqref{eq:qgghzr} satisfies the above property.
\end{proposition}

\begin{proof}
Here we present the proof only for the four-partite case; the extension to general $\mathtt{q}$ is straightforward.
The converse is immediate.
We now prove the nontrivial direction.
\\

\noindent
Let us relabel the four parties by $A$, $B$, $C$, and $D$. Since $\rho_{ABC}$ is fully separable, its support is spanned by fully product vectors.
Let
\begin{equation}
R=\rank(\rho_{ABC})=\rank(\rho_D)\ .
\end{equation}
Choose $R$ linearly independent fully product vectors
\begin{equation}
|a_j b_j c_j\rangle
:=
|a_j\rangle\otimes|b_j\rangle\otimes|c_j\rangle\ ,
\quad
j=1,\dots,R
\end{equation}
spanning $\supp(\rho_{ABC})$.
Since $|\Psi\rangle$ is a purification of $\rho_{ABC}$ and $\rank(\rho_D)=R$, we can write it as
\begin{equation}\label{eq:4partite-decomp}
|\Psi\rangle
=
\sum_{j=1}^R
|a_j b_j c_j\rangle |d_j\rangle \ ,
\end{equation}
where the vectors
\bea
\{|d_j\rangle\}_{j=1}^R\subset\mathcal H_D
\eea
are linearly independent.
\\

\noindent
Tracing out $A$ in Eq.~\eqref{eq:4partite-decomp}, we obtain
\begin{equation}\label{eq:rhoBCD-4partite}
\rho_{BCD}
=
\sum_{j,k=1}^R
\langle a_k|a_j\rangle\,
|b_j c_j d_j\rangle\langle b_k c_k d_k| \ .
\end{equation}
Since $\rho_{BCD}$ is fully separable, it is separable across $BC|D$.
Since the vectors $\{|d_j\rangle\}_{j=1}^R$ are linearly independent, there exists a linear map
\begin{equation}
M_D:\ \mathcal H_D\to \mathbb C^2
\end{equation}
such that
\begin{equation}
M_D|d_j\rangle=|0\rangle\ ,\quad
M_D|d_k\rangle=|1\rangle\ ,\quad
M_D|d_l\rangle=0
\end{equation}
for $l\neq j,k$.
Applying the local filter $I_{BC}\otimes M_D$ to $\rho_{BCD}$, we obtain
\begin{equation}
\tilde\rho_{BCD}
=
(I_{BC}\otimes M_D)\rho_{BCD}(I_{BC}\otimes M_D^\dagger)\ ,
\end{equation}
which is still separable across $BC|D$. Using Eq.~\eqref{eq:rhoBCD-4partite}, we find
\begin{align}
\tilde\rho_{BCD}
&=
\langle a_j|a_j\rangle\,|b_j c_j\rangle\langle b_j c_j|\otimes |0\rangle\langle 0|
+
\langle a_k|a_k\rangle\,|b_k c_k\rangle\langle b_k c_k|\otimes |1\rangle\langle 1|
\nonumber\\
&\quad
+
\langle a_k|a_j\rangle\,|b_j c_j\rangle\langle b_k c_k|\otimes |0\rangle\langle 1|
+
\langle a_j|a_k\rangle\,|b_k c_k\rangle\langle b_j c_j|\otimes |1\rangle\langle 0| \ .
\end{align}
If $|b_j c_j\rangle$ and $|b_k c_k\rangle$ are linearly independent, then $\tilde\rho_{BCD}$ is of the form~\eqref{eq:2term-lemma-state} across $BC|D$.
Hence, by Lemma~\ref{lemma:sep},
\begin{equation}\label{eq:ai-nonzero-implies-bc-factors-short}
\langle a_j|a_k\rangle\neq 0
\quad\Longrightarrow\quad
|b_k c_k\rangle \propto |b_j c_j\rangle
\quad\Longrightarrow\quad
|b_k\rangle\propto |b_j\rangle\ ,\ |c_k\rangle\propto |c_j\rangle \ .
\end{equation}
Cyclically permuting $A,B,C$, we likewise obtain
\begin{equation}\label{eq:bi-nonzero-implies-ac-factors-short}
\langle b_j|b_k\rangle\neq 0
\quad\Longrightarrow\quad
|c_k\rangle\propto |c_j\rangle \ ,\ |a_k\rangle\propto |a_j\rangle\ ,
\end{equation}
and
\begin{equation}\label{eq:ci-nonzero-implies-ab-factors-short}
\langle c_j|c_k\rangle\neq 0
\quad\Longrightarrow\quad
|a_k\rangle\propto |a_j\rangle\ ,\ |b_k\rangle\propto |b_j\rangle \ .
\end{equation}
\\

\noindent
We now show that $\{|a_j\rangle\}$, $\{|b_j\rangle\}$, and $\{|c_j\rangle\}$ are pairwise orthogonal.
If $\langle a_j|a_k\rangle\neq 0$, then Eq.~\eqref{eq:ai-nonzero-implies-bc-factors-short} implies $|b_k\rangle\propto|b_j\rangle$ and $|c_k\rangle\propto|c_j\rangle$, thus $\langle b_j|b_k\rangle\neq 0$, and then Eq.~\eqref{eq:bi-nonzero-implies-ac-factors-short} gives $|a_k\rangle\propto|a_j\rangle$. Thus,
\bea
|a_k b_k c_k\rangle \propto |a_j b_j c_j\rangle \ ,
\eea
contradicting the linear independence of $\{|a_j b_j c_j\rangle\}_{j=1}^R$.
By the same logic, starting from
$\langle b_j|b_k\rangle\neq 0$ or $\langle c_j|c_k\rangle\neq 0$, we also get
$\langle b_j|b_k\rangle=\langle c_j|c_k\rangle=0$ for $j\neq k$.
Therefore, we get
\begin{equation}\label{eq:abc-orthogonal}
\langle a_j|a_k\rangle=\langle b_j|b_k\rangle=\langle c_j|c_k\rangle=0
\qquad
(j\neq k) \ .
\end{equation}
Tracing out $D$ in Eq.~\eqref{eq:4partite-decomp}, we get
\begin{equation}
\rho_{ABC}
=
\sum_{j,k=1}^R
\langle d_k|d_j\rangle\,|a_j b_j c_j\rangle\langle a_k b_k c_k| \ .
\end{equation}
Since $\rho_{ABC}$ is fully separable, it is separable across $A|BC$.
For any $j\neq k$, projecting $\rho_{ABC}$ onto
\begin{equation}
\Span\{|a_j\rangle,|a_k\rangle\}\otimes \Span\{|b_j c_j\rangle,|b_k c_k\rangle\}
\end{equation}
and using Lemma~\ref{lemma:sep}, we find
\begin{equation}\label{eq:d-orthogonal-short}
\langle d_k|d_j\rangle=0
\qquad
(k\neq j) \ .
\end{equation}
In summary, each of the four families
\begin{equation}\label{eq:4pairs}
\{|a_j\rangle\}_{j=1}^R\ ,\quad
\{|b_j\rangle\}_{j=1}^R\ ,\quad
\{|c_j\rangle\}_{j=1}^R\ ,\quad
\{|d_j\rangle\}_{j=1}^R
\end{equation}
is pairwise orthogonal.
Therefore, there exist local unitaries sending them to the standard basis, and
\begin{equation}
(U_A\otimes U_B\otimes U_C\otimes U_D)|\Psi\rangle
=
\sum_{j=0}^{R-1}\lambda_j |jjjj\rangle \ ,
\end{equation}
with $\sum_{j=0}^{R-1}|\lambda_j|^2=1$. This is precisely the generalized GHZ form.
\end{proof}

\noindent
This proposition motivates a broader question: how far can separability-based characterizations of multipartite entanglement be extended beyond finite-dimensional systems?
Quantum field theory (QFT) provides a natural but much more difficult arena for this question.
The theory has infinitely many degrees of freedom, and ultraviolet (UV) divergences further complicate its entanglement structure and require regularization.
Moreover, entanglement is generically present in spatial subregions~\cite{Witten:2018zxz}.
Therefore, this result implies that the local QFT for the spatial partition does not admit a GHZ-type state if the qudit state can approach the QFT.\footnote{It has been shown in Ref.~\cite{Balasubramanian:2025hxg} that holographic states cannot be pure GHZ-states.}
GHZ-type states can only explore mode-wise entanglement, as in Lifshitz's vacuum state.
In the next section, we will demonstrate computability via third-order negativity in CFT$_2$.
\section{Illustration in CFT}
\label{sec:3rd-nega-CFT}
\noindent
In this section, we illustrate the field-theoretic realization of the third-order negativity $I_5$ by evaluating it in a two-dimensional CFT.
We consider the vacuum $\Psi$ on the infinite spatial line and choose three ordered points $z_1<z_2<z_3$ on the same Euclidean-time slice, which define the spatial regions
\begin{equation}\label{eq:partition}
A = (-\infty, z_1]\cup[z_3,+\infty)\ ,\quad B = [z_1,z_2]\ ,\quad C = [z_2,z_3]
\end{equation}
as shown in the upper-left part of Fig.~\ref{fig:CFT-replica}.
The intervals $B$ and $C$ are adjacent, while $A$ is the complement
of $B\cup C$.
We introduce an ultraviolet regulator of length $\epsilon$, so that the reduced density matrix and its partial transpose are well defined. With
\begin{equation}
\rho_{BC}=\mathrm{Tr}_A|\Psi_{ABC}\rangle\langle \Psi_{ABC}|\ ,
\end{equation}
the quantity of interest is
\begin{equation}
I_5=\mathrm{Tr}\big[(\rho_{BC}^{\Gamma})^3\big] \ ,
\end{equation}
where $\Gamma$ denotes the partial transpose on $C$.
\\

\begin{figure}[htbp]
\centering
\includegraphics[scale = 0.9]{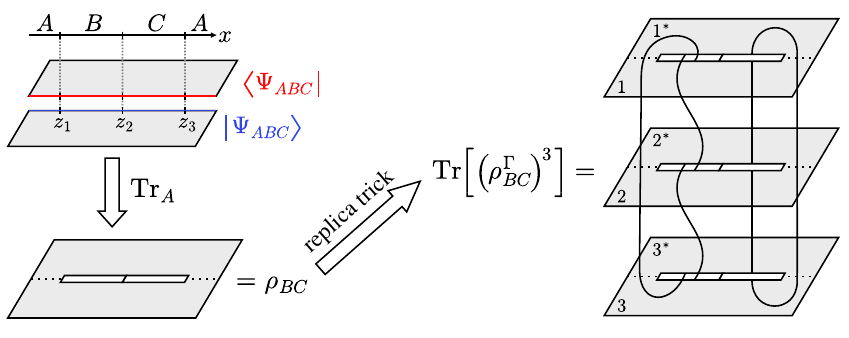}
\caption{
Path-integral construction of $I_5=\Tr[(\rho_{BC}^{\Gamma})^3]$.
Upper left: the vacuum ket (blue) and bra (red), together with the spatial partition~\eqref{eq:partition}. 
Lower left: tracing over $A$ sews the bra and ket edges of the cut along $A$ (dashed edges), leaving the cuts along $B\cup C$ open and producing $\rho_{BC}$.
Right: the three replicas are sewn cyclically along $B$ and in the opposite cyclic order along $C$, thereby implementing the partial transpose.
The curves indicate the identifications of the cut edges. }
\label{fig:CFT-replica}
\end{figure}

\noindent
The vacuum ket and bra are prepared by Euclidean path integrals over the lower and upper half-planes, respectively (Fig.~\ref{fig:CFT-replica}, Upper left).
Tracing over $A$ sews the corresponding edges of the cut along $A$, leaving the cuts along $B\cup C$ open (Fig.~\ref{fig:CFT-replica}, Lower left).
The partial transpose with respect to $C$ exchanges the upper and lower edges of the cut on $C$.
Taking the third power, then the trace, sews the three replicas cyclically along $B$ and in the opposite cyclic order along $C$ (Fig.~\ref{fig:CFT-replica}, Right).
\\

\noindent
Writing $g=(123)$, the sewing permutations are therefore
\begin{equation}
\Omega_A=\mathrm{id}\ , \quad \Omega_B=g\ , \quad \Omega_C=g^{-1}\ ,
\end{equation}
which are precisely the replica permutations appearing in the replica representation of $I_5$ in Eq.~\eqref{eq:I5-replica}. 
The twist permutation at each endpoint records the change in circling that endpoint:
\begin{equation}\label{eq:twist-g}
g_1 = \Omega_A^{-1} \Omega_B = g\ , \quad
g_2 = \Omega_B^{-1} \Omega_C = g^{-2} =g\ , \quad
g_3 = \Omega_C^{-1} \Omega_A = g\ .
\end{equation}
Here $g^{-2}=g$ follows from $g^3=\mathrm{id}$.
Thus, all three insertions carry the same three-cycle.
Moreover, since the three-cycle acts transitively on the three replicas, the corresponding replica surface is connected.
As shown explicitly in Appendix~\ref{app:twist-ope}, its six bra and ket half-planes can be uniformized into six equilateral triangles, which form a hexagon with opposite sides identified (Fig.~\ref{fig:map}).
The replica surface is therefore a torus.
\\

\noindent
In the replicated CFT, each endpoint is represented by a local twist operator $\sigma$ that implements the corresponding monodromy.
For a twist operator associated with a cycle of length $k$, its holomorphic and antiholomorphic weights are~\cite{Lunin:2000yv}
\begin{equation}
h=\bar{h}=\frac{c}{24}\left(k-\frac{1}{k}\right)\ ,
\end{equation}
where $c$ is the central charge of a single copy of the original CFT.
For $g = (123)$, we have $k=3$, which gives
\begin{equation}
h_g=\bar{h}_g=\frac{c}{24}\bigg(3-\frac{1}{3}\bigg)=\frac{c}{9}\ .
\end{equation}
Hence, all three twist operators have total scaling dimension
\begin{equation}\label{eq:twist-dim}
\Delta_1=\Delta_2=\Delta_3=\Delta_g = h_g+\bar{h}_g=\frac{2c}{9}\ .
\end{equation}
The replica path integral can be expressed as the three-point function of these twist operators,
\begin{equation}\label{eq:neg3-3pt}
\Tr\big[(\rho_{BC}^{\Gamma})^3\big]
=
\big\langle
\sigma_{g_1}(z_1,\bar{z}_1)\,
\sigma_{g_2}(z_2,\bar{z}_2)\,
\sigma_{g_3}(z_3,\bar{z}_3)
\big\rangle\ .
\end{equation}
\\

\noindent
Conformal symmetry fixes the three-point function up to an OPE coefficient $C_{123}$,
\begin{equation}\label{eq:3pt-general}
\big\langle
\sigma_1(z_1,\bar{z}_1)\,
\sigma_2(z_2,\bar{z}_2)\,
\sigma_3(z_3,\bar{z}_3)
\big\rangle
=
\frac{C_{123}}
{|z_{12}/\epsilon|^{\Delta_1+\Delta_2-\Delta_3}
 |z_{23}/\epsilon|^{\Delta_2+\Delta_3-\Delta_1}
 |z_{31}/\epsilon|^{\Delta_3+\Delta_1-\Delta_2}} \ ,
\end{equation}
where $z_{jk}:=z_j-z_k$.
Conformal symmetry determines the position dependence, whereas $C_{123}$ contains additional information about the CFT.
Appendix~\ref{app:twist-ope} uniformizes the connected three-sheeted
replica surface and derives the exact relation
\begin{equation}
C_{123}=3^{-\frac{c}{2}}Z_{\mathcal{C}}(\tau_*, \bar{\tau}_*)\ , \quad
\tau_*=e^{i\frac{\pi}{3}}\ ,
\end{equation}
where $Z_{\mathcal{C}}$ is the torus partition function of a single copy
of the original CFT $\mathcal{C}$.
The factor $3^{-c/2}$ arises from the Weyl anomaly associated with the covering map. 
The remaining theory dependence is encoded in $Z_{\mathcal{C}}$.
\\

\noindent
For a holographic CFT in the semiclassical Einstein-gravity regime, assuming the torus partition function is dominated at leading order by the usual solid-torus saddles,
\begin{equation}
\log Z_{\cal C}(\tau_*, \bar{\tau}_*)=\frac{\pi\sqrt{3}}{12}c+{\cal O}(c^0) \ .
\end{equation}
It follows that
\begin{equation}
\log C_{123} = c\bigg(\frac{\pi\sqrt{3}}{12}-\frac{1}{2}\log 3\bigg) + \mathcal{O}(c^0)\ .
\label{eq:C123-holographic}
\end{equation}
Combining Eqs.~\eqref{eq:3pt-general} and~\eqref{eq:C123-holographic}, we arrive at
\begin{equation}
\log\Tr\big[(\rho_{BC}^\Gamma)^3\big] = -\frac{c}{9}\log\bigg(\frac{|z_1 - z_2|^2|z_2 - z_3|^2|z_3 - z_1|^2}{\epsilon^6}\bigg) + c\bigg(\frac{\pi\sqrt{3}}{12}-\frac{1}{2}\log 3\bigg)+{\cal O}(c^0) \ .
\end{equation}
\\

\noindent
This calculation provides a continuum realization of the tripartite replica invariant studied in the preceding sections.
Its interval dependence follows from conformal symmetry, while its normalized three-point coefficient is determined by the partition function on a genus-one covering surface.
The holographic approximation evaluates this coefficient at leading order in large $c$.

\section{Discussion and Conclusion}
\label{sec:discussion}
\noindent
This paper studied multipartite entanglement by exploring the interplay between separability criteria and multi-invariants.
We showed that the third-order negativity $I_5$~\eqref{eq:I5-intro} satisfies
\bea
I_5=1
\eea
if and only if a tripartite pure state is a fully product state.
Its convex-roof extension~\cite{Eisert:2001xyp} then provides a necessary and sufficient criterion for full separability of tripartite mixed states. This criterion naturally applies to four-qubit pure states and yields four ``$\phi$-type'' tripartite measures.
Together with the six two-tangles~\cite{Bennett:1996gf,Hill:1997pfa,Wootters:1997id}, four three-tangles~\cite{Coffman:1999jd}, and four quadripartite measures, these quantities form a set of eighteen independent entanglement measures, matching the number of LU-invariant degrees of freedom of a four-qubit pure state.
For three-qubit pure states, the upper bound of $\phi_{ABC}$ and the lower bound of $I_5$ are saturated by the GHZ- and W-states, respectively, thereby interpolating between separable states and distinct classes of entangled states.
We further proved that the four-tangle~\cite{Uhlmann:2000ckp,Wong:2000cmz} is the only entanglement measure that can remain nonzero while all other seventeen vanish, uniquely characterizing generalized GHZ-type four-qubit states. Beyond qubits, we developed separability criteria for general finite-dimensional multipartite systems based on multi-invariants.
We showed that a $\mathtt{q}$-partite pure state is LU-equivalent to a generalized GHZ-state if and only if all of its ($\mathtt{q}-1$)-partite reduced density matrices are fully separable.
This establishes a direct link between multipartite separability and global entanglement structure.
Finally, we discussed applications to QFT by computing the third-order negativity in CFT$_2$, highlighting the connection between multipartite measures and field-theoretic observables.
\\

\noindent
There are several promising directions for future investigation. Although the convex-roof constructions used here are conceptually natural, they are generally difficult to evaluate in practice. Developing analytic formulae or efficient numerical methods for entanglement measures such as $\phi_{ABC}$~\cite{Oreshkov:2006gqe} would significantly enhance the applicability of the separability criteria.
\\

\noindent
Following Table~\ref{tab:hierarchy-measures}, one might expect that for a five-qubit pure state, there are $\binom{5}{2} = 10$ bipartite, $2\times\binom{5}{3} = 20$ tripartite, and $4\times\binom{5}{4} = 20$ quadripartite entanglement measures.
However, this already exceeds the number of independent parameters of a five-qubit pure state~\cite{Carteret:2000jop}:
\begin{equation}
10 + 20 + 20 > 47\ .
\end{equation}
A similar mismatch occurs for six qubits,
\begin{equation}
\binom{6}{2} + 2\times\binom{6}{3} + 4\times\binom{6}{4} > 2^{6 + 1} - 2 - 3\times 6\ .
\end{equation}
This count indicates that the four-qubit system is the last case in which the naive hierarchical organization remains compatible with parameter independence.
For five qubits and beyond, multipartite measures must satisfy nontrivial constraints and hidden relations. Identifying these relations and extracting an independent minimal set of measures are important open problems.
\\

\noindent
In this work, we express almost all four-qubit entanglement measures in terms of reduced density matrices, with the notable exception of the hyperdeterminant.
Our results indicate that the hyperdeterminant is nonvanishing only for specific classes of genuinely multipartite entangled states.
It would therefore be valuable to express the hyperdeterminant directly in terms of reduced density matrices, as this could provide deeper insight into the structure of four-qubit entanglement.
\\

\noindent
More generally, many-body systems exhibit qualitatively different types of entanglement. It is therefore important to develop measures—such as $\phi_{ABC}$—that can detect how a separable state transitions into distinct classes of entangled states.
Since any valid entanglement measure must vanish on separable states, the study of separability provides a natural guiding principle for constructing such measures.
Moreover, our framework applies to finite-dimensional Hilbert spaces.
It is directly accessible to current experimental platforms, particularly in the four-qubit pure-state setting, which represents the simplest realization of tripartite mixed states.
\\

\noindent
Finally, it is important to explore how these ideas extend to QFT.
A common approach is to approximate a QFT by qudit systems in the large-local-dimension limit.
However, our results suggest a limitation of this correspondence: generalized GHZ-type qudit states are separable under all reduced density matrices, whereas local QFT does not admit strictly separable states for spatial subregions~\cite{Witten:2018zxz}.
This indicates that not all entanglement structures realized in qudit systems are compatible with local QFT.
It would be interesting to characterize which classes of states remain consistent with the intrinsic entanglement structure of QFT.
This question is closely related to the fact that, while qudit systems admit a standard partial trace, local QFT lacks a straightforward tensor-factor decomposition for spatial regions~\cite{Witten:2018zxz}.
Consequently, qudit systems may more naturally capture mode-wise entanglement, rather than spatial entanglement in local QFT.

\section*{Acknowledgments}
\noindent 
We want to express our gratitude to Chong-Sun Chu, Ruihua Fan, Jiayu Ran, Peng Tan, and Huangjun Zhu for their helpful discussions. 
CTM thanks Nan-Peng Ma for his encouragement. 
\appendix
\section{\boldmath Lower Bound of $I_5$\unboldmath}
\label{appx:I5-lower}
\begin{proposition}\label{prop:I5-lower-bound}
For any three-qubit pure state, $I_5\ge 2/9$, and the lower bound is saturated if and only if the state is LU-equivalent to the $W$-state $(|100\rangle+|010\rangle+|001\rangle)/\sqrt{3}$.
\end{proposition}
\begin{proof}
Any three-qubit pure state is LU-equivalent to the following canonical form~\cite{Acin:2000jx,Brun:2001gyy}
\begin{equation}\label{eq:3-canonical}
|\psi\rangle
=
\lambda_0 |000\rangle
+\lambda_1 e^{i\varphi}|100\rangle
+\lambda_2 |101\rangle
+\lambda_3 |110\rangle
+\lambda_4 |111\rangle \ ,
\end{equation}
with $\lambda_i\ge 0$, $0\le \varphi <\pi$, and $\sum_{j=0}^4 \lambda_j^2=1$. Let
\begin{equation}\label{eq:abcde}
a=\lambda_0^2\ ,\quad
b=\lambda_1^2\ ,\quad
c=\lambda_2^2\ ,\quad
d=\lambda_3^2\ ,\quad
e=\lambda_4^2\ ,
\end{equation}
so that $a,b,c,d,e\ge 0$ and $a+b+c+d+e=1$.
Substituting Eq.~\eqref{eq:3-canonical} directly into the definition of $I_5$~\eqref{eq:I5-intro} gives
\begin{equation}\label{eq:I5-acin-full}
\begin{aligned}
I_5
=&\,
a^3+b^3+c^3+d^3+e^3\\
&+3(a^2b+ab^2+b^2c+bc^2+b^2d+bd^2+c^2e+ce^2+d^2e+de^2)\\
&+3(abc+abd+acd+bcd+bce+bde+cde)
\\
&+6\sqrt{bcde}\,(b+c+d+e)\cos\varphi \ .
\end{aligned}
\end{equation}
Since we have $b+c+d+e=1-a\ge 0$, the last term in Eq.~\eqref{eq:I5-acin-full} is minimized, at fixed $a,b,c,d,e$, by taking $\cos\varphi=-1$.
Now define
\begin{equation}\label{eq:sy}
s:=c+d\ ,\quad y:=\sqrt{cd}\ ,
\end{equation}
so that
\begin{equation}
0\le y\le \frac{s}{2}\ ,
\quad
e=1-a-b-s \ .
\end{equation}
Using Eq.~\eqref{eq:sy}, one finds after simplification
\begin{equation}
I_5
\ge
1+3(a + b)^2+3bs-3a-3b
+3(2a-1)y^2
-6(1-a)\sqrt{be}\,y \ .
\end{equation}
Thus, the problem is reduced to minimizing
\begin{equation}\label{eq:I5-asy}
\mathcal{I}_5 \equiv 1+3(a + b)^2+3bs-3a-3b
+3(2a-1)y^2
-6(1-a)\sqrt{be}\,y
\end{equation}
in the compact domain:
\begin{equation}\label{eq:domain-asy}
a,b,s,y\ge 0\ ,\quad
a+b+s\le 1\ ,\quad
0\le y\le \frac{s}{2}\ .
\end{equation}
\\

\noindent
We first exclude interior critical points.
Assuming
\begin{equation}
a>0\ ,\quad b>0\ ,\quad s>0\ ,\quad e>0\ ,\quad 0<y<\frac{s}{2}\ ,
\end{equation}
we then get
\begin{equation}
\frac{\partial \mathcal{I}_5}{\partial s}
=
3b+\frac{3(1-a)b}{\sqrt{be}}y> 0 \ .
\end{equation}
Therefore, there is no interior critical point.
Hence, the global minimum must lie on the boundary of the domain~\eqref{eq:domain-asy}.
We now analyze the boundary faces.
\\
\noindent
\boldmath \textbf{(i)} $y=0$.\unboldmath\ Then Eq.~\eqref{eq:I5-asy} reduces to
\begin{equation}
\mathcal{I}_5
=
1+3(a+b)^2-3(a+b)+3bs \ge 3(a+b)^2-3(a+b)+1\ge \frac{1}{4}>\frac{2}{9}\ .
\end{equation}
\\
\noindent
\boldmath \textbf{(ii)} $a=0$.\unboldmath\ Then Eq.~\eqref{eq:I5-asy} becomes
\begin{equation}
\mathcal{I}_5
=
1+3b^2+3bs-3b-3y^2-6\sqrt{be}\,y \ .
\end{equation}
Since the coefficients of $y^2$ and $y$ are both negative, the minimum occurs at the largest allowed value, $y=s/2$.
Hence, we get
\begin{equation}
\mathcal{I}_5
=
1+3b^2+3bs-3b-\frac{3}{4} s^2-3s\sqrt{be} \ .
\end{equation}
Using
\bea
\sqrt{be}\le \frac{b+e}{2} = \frac{1-s}{2} \ ,
\eea
we then get
\begin{equation}
\mathcal{I}_5
\ge
1+3b^2+3bs-3b-\frac{3}{4} s^2-\frac{3}{2} s(1-s)=
\frac{1}{4}+3\left(b-\frac{1-s}{2}\right)^2 \ge \frac{1}{4}>\frac{2}{9} \ .
\end{equation}
\\
\noindent
\boldmath \textbf{(iii) $b=0$}.\unboldmath\ Then Eq.~\eqref{eq:I5-asy} gives
\begin{equation}
\mathcal{I}_5
=
1-3a+3a^2+3(2a-1)y^2 \ .
\end{equation}
If $a\ge 1/2$, then $2a-1\ge 0$, so the minimum occurs at $y=0$, and hence we get
\begin{equation}
\mathcal{I}_5\ge 1-3a+3a^2\ge \frac{1}{4} \ .
\end{equation}
If $a\le 1/2$, then $2a-1\le 0$, so the minimum occurs at the largest allowed value $y=s/2$. 
Since $s\le 1-a$ as well, the minimum is attained at $s=1-a$.
Therefore, we obtain
\begin{equation}
\mathcal{I}_5
\ge
1-3a+3a^2+\frac{3}{4}(2a-1)(1-a)^2
=
\frac{6a^3-3a^2+1}{4} \ .
\end{equation}
On $0\le a\le 1/2$, the RHS is minimized at $a=1/3$, giving
\begin{equation}
\mathcal{I}_5\ge \frac{2}{9} \ .
\end{equation}
Equality holds if and only if
\bea
a=\frac{1}{3}\ , \quad s=\frac{2}{3}\ , \quad y=\frac{1}{3} \ ,
\eea
which implies
\begin{equation}\label{eq:W-condition}
a=c=d=\frac{1}{3}\ ,\quad b=e=0 \ .
\end{equation}
\\
\noindent
\boldmath \textbf{(iv) $e=0$}.\unboldmath\ Then Eq.~\eqref{eq:I5-asy} reduces to
\begin{equation}
\mathcal{I}_5
=
1-3as+3(2a-1)y^2 \ .
\end{equation}
If $a\ge 1/2$, then the minimum is at $y=0$, and since $s\le 1-a$,
\begin{equation}
\mathcal{I}_5\ge 1-3a(1-a)\ge \frac{1}{4} \ .
\end{equation}
If $a\le 1/2$, then the minimum is at $y=s/2$, and again $s\le 1-a$, so
\begin{equation}
\mathcal{I}_5
\ge
1-3a(1-a)+\frac{3}{4}(2a-1)(1-a)^2
=
\frac{6a^3-3a^2+1}{4}\ge \frac{2}{9} \ .
\end{equation}
Equality holds if and only if
\bea
a=\frac{1}{3}\ , \quad s=\frac{2}{3}\ , \quad y=\frac{1}{3} \ ,
\eea
which again implies Eq.~\eqref{eq:W-condition}.
\\
\noindent
\boldmath \textbf{(v)} $s=0$.\unboldmath\ Then necessarily $y=0$, so this falls under case (i).
\\

\noindent
\boldmath \textbf{(vi) $y=s/2$}.\unboldmath\ Define $u := \sqrt{be} + s/2$. 
Since $a+b+s+e=1$, we have
\begin{equation}
    1-a-2u=(\sqrt b-\sqrt e)^2\geq0\ ,
\end{equation}
and hence $0\leq u\leq1/2$.
The endpoint $u=1/2$ implies $a=0$ and is already covered by case (ii).
For $u<1/2$, Eq.~\eqref{eq:I5-asy} can be rearranged as
\begin{equation}
\begin{split}
    \mathcal I_5-\frac{2}{9}     &=
    \frac{(3u-1)^2(6u+1)}{36}\\
&+
    \frac{3}{1-2u}
    \left[
        \bigg(
            (1-u)a-\frac{(1+u)(1-2u)}{2}
        \bigg)^2
        +a\left((1-u)\sqrt b-u\sqrt e\right)^2
    \right]
    \geq0\ . 
\end{split}
\end{equation}
Equality requires $u=a=1/3$.
Then $1-a-2u=0$ gives $b=e$, while the last square requires $2\sqrt b=\sqrt e$.
Thus $b=e=0$ and $s=2/3$, reproducing Eq.~\eqref{eq:W-condition}.
\\

\noindent
Combining all boundary cases, we conclude that
\begin{equation}
I_5\ge \frac{2}{9} \ .
\end{equation}
Moreover, the equality conditions in all cases reduce to Eq.~\eqref{eq:W-condition}.
This is precisely the W-state in the canonical form.
Therefore, the lower bound is saturated if and only if the state is LU-equivalent to the W-state.
\end{proof}
\section{\boldmath Upper Bound of $\phi_{ABC}$\unboldmath}
\label{appx:phi-upper}
\begin{proposition}\label{prop:phi-upper-bound}
For any three-qubit pure state, $\phi_{ABC}\le 99/2$, and the upper bound is saturated if and only if the state is LU-equivalent to the GHZ-state
$(|000\rangle+|111\rangle)/\sqrt{2}$.
\end{proposition}
\begin{proof}
Using the conventions in Eqs.~\eqref{eq:3-canonical} and~\eqref{eq:abcde}, a direct substitution into the definition of $\phi_{ABC}$~\eqref{eq:phiABC} gives
\begin{equation}\label{eq:phi-abc-acin}
\frac{\phi_{ABC}}{18}
=
11a(1-a)
-11ab
+8be
-3a(c+d)
+(8 - 4a)cd
-4(4-a)\sqrt{bcde}\cos\varphi\ .
\end{equation}
Since $4-a>0$, the RHS is maximized, at fixed $a,b,c,d,e$, by taking $\cos\varphi=-1$.
Hence, we get
\begin{equation}\label{eq:phi-phase-max}
\frac{\phi_{ABC}}{18}
\le
11a(1-a)
-11ab
+8be
-3a(c+d)
+(8 - 4a)cd
+4(4-a)\sqrt{bcde}\ .
\end{equation}
Next, fix $a,b,e$ and $s:=c+d$. The dependence on $c,d$ in the RHS of Eq.~\eqref{eq:phi-phase-max} is
\begin{equation}
-3as + (8-4a)cd + 4(4-a)\sqrt{becd} \ .
\end{equation}
Because $0\le a\le 1$, both coefficients $(8-4a)$ and $4(4-a)\sqrt{be}$ are positive, so this is increasing in $cd$.
Since $cd\le s^2/4$ with the equality held at $c=d=s/2$, the maximum is attained when $c=d=s/2$.
Therefore, it suffices to maximize
\begin{equation}\label{eq:Phi-abs}
\Phi(a,b,s)\equiv
11a(1-a)
-11ab
+8be
-3as
+(2-a)s^2
+2(4-a)s\sqrt{be} \ ,
\end{equation}
where now
\begin{equation}
a + b + s + e = 1\ ,
\quad
a,b,s,e\ge 0 \ .
\end{equation}
\\

\noindent
We now analyze the boundary faces of the simplex.
\\
\noindent
\boldmath \textbf{(i)} $a=0$.\unboldmath\ Then
\begin{equation}
\Phi(0,b,s)=2s^2+8be+8s\sqrt{be} \ ,
\quad b+s+e=1 \ .
\end{equation}
Using $be\le (1-s)^2/4$ and $\sqrt{be}\le (1-s)/2$, we get
\begin{equation}
\Phi(0,b,s)\le 2s^2+2(1-s)^2+4s(1-s)=2 \ .
\end{equation}
\\
\noindent
\boldmath \textbf{(ii)} $e=0$.\unboldmath\ Then $a+b+s=1$ and
\begin{equation}
\Phi(a,b,s)=11a(1-a)-11ab-3as+(2-a)s^2
=2s^2+8as-as^2\ .
\end{equation}
For fixed $s$, this is linear in $a$, with positive coefficient $s(8-s)$ on $[0,1]$, so the maximum occurs at $a=1-s$, i.e. $b=0$.
Hence, we get
\begin{equation}
\Phi\le s^3-7s^2+8s \le \frac{68}{27} \ .
\end{equation}
\\
\noindent
\boldmath \textbf{(iii)} $b=0$.\unboldmath\ Then
\begin{equation}
\Phi(a,0,s)=11a(1-a)-3as+(2-a)s^2 \ .
\end{equation}
For fixed $a$, this is convex in $s$, since $2-a\ge 1>0$.
Hence, the maximum on $0\le s\le 1-a$ occurs at an endpoint:
\begin{align}
\Phi(a,0,0)&=11a(1-a)\le \frac{11}{4} \ ,\\
\Phi(a,0,1-a)&=(1-a)(a^2+5a+2)\le \frac{68}{27} \ .
\end{align}
Thus, we obtain
\begin{equation}
\max_{b=0}\Phi=\frac{11}{4} \ .
\end{equation}
\\
\noindent
\boldmath \textbf{(iv)} $s=0$.\unboldmath\ Then
\begin{equation}
\Phi(a,b,0)=11a(1-a)-11ab+8b(1-a-b) \ .
\end{equation}
For fixed $a$, this is concave in $b$, and its maximum is:
\begin{equation}
\max_b \Phi(a,b,0)=
\begin{cases}
\dfrac{(8+3a)^2}{32} \ , & 0\le a\le \dfrac{8}{19}\ ,\\[6pt]
11a(1-a) \ , & \dfrac{8}{19}\le a\le 1\ .
\end{cases}
\end{equation}
Therefore, we obtain
\begin{equation}
\max_{s=0}\Phi=\frac{11}{4} \ ,
\end{equation}
attained at $a=1/2$, $b=0$.
Hence, every boundary face satisfies
\begin{equation}\label{eq:boundary-bound}
\Phi\le \frac{11}{4} \ ,
\end{equation}
and equality can only occur at
\begin{equation}
a=\frac{1}{2}\ ,\quad b=0\ ,\quad s=0\ ,\quad e=\frac{1}{2}\ .
\end{equation}
\\

\noindent
We still need to exclude interior critical points.
For fixed $a,s$, one computes
\begin{equation}\label{eq:Phi-b-concave}
\frac{\partial^2 \Phi}{\partial b^2}
=
-16-\frac{(4-a)s(1-a-s)^2}{2[b(1-a-s-b)]^{3/2}}<0
\qquad
(0<b<1-a-s)\ .
\end{equation}
Thus, for fixed $(a,s)$, the function is strictly concave in $b$, so any interior critical point is unique in the $b$-direction.
Now assume:
\begin{equation}
a>0\ ,\quad b>0\ ,\quad s>0\ ,\quad e>0\ .
\end{equation}
Let $r:=\sqrt{b/e}>0$, so that
\begin{equation}
b=\frac{(1-a-s)r^2}{1+r^2}\ .
\end{equation}
Rewrite $\Phi(a,b,s)$ in terms of $a,s,r$, which we denote as $\tilde{\Phi}(a,r,s)$.
Then the stationarity equations $\partial_a\tilde{\Phi}=\partial_r\tilde{\Phi}=\partial_s\tilde{\Phi}=0$ become polynomial equations in $(a,r,s)$:
\begin{align}
\partial_a \tilde{\Phi} = 0\ \Rightarrow\ &(1-r)^2(1+r^2)s^2 + \big((5+6a)r^2 + 22a - 11\big)\nonumber\\
&\qquad - (8r^4 + (4a - 10)r^3 + 21r^2 + (4a - 10)r - 3)s=0\ ,
\label{eq:pda}\\
\partial_r \tilde{\Phi} = 0\ \Rightarrow\ &\big((4 - a)r^4 - 8r^3 + 8r - (4 - a)\big)s + (3a+8)r^3 + 19ar - 8r = 0\ ,
\label{eq:pdr}\\
\partial_s \tilde{\Phi} = 0\ \Rightarrow\ &\big(4(1 - r)^4 - 2a(1-r)^2(1+r^2)\big)s + 8r(1-r)^2\nonumber\\
&\qquad+ 2a^2 r(1+r^2) + a(8r^4 - 10r^3 + 21r^2 - 10r - 3) = 0\ .\label{eq:pds}
\end{align}
Eliminating $s$ from Eqs.~\eqref{eq:pdr} and~\eqref{eq:pds} yields
\begin{equation}\label{eq:q-ar}
2a^2r(r+1)+a(8r^3-16r^2+17r-3)-32r^3+76r^2-56r+12=0 \ .
\end{equation}
Eliminating $s$ from Eqs.~\eqref{eq:pda} and~\eqref{eq:pdr}, and then eliminating $a$ between the resulting equation and Eq.~\eqref{eq:q-ar}, one obtains the resultant
\begin{equation}\label{eq:r-resultant}
(r-1)^4(r+1)(2r-1)(4r-7)(4r^2-4r+3)(5r^3-r^2+21r-17)^2=0 \ .
\end{equation}
Hence, the only positive candidates are:
\begin{equation}
r=\frac{1}{2}\ ,\quad r=1\ ,\quad r=\frac{7}{4}\ ,
\quad
r=r_*\ ,
\end{equation}
where $r_*$ is the unique positive root of
\begin{equation}
5r^3-r^2+21r-17=0 \ .
\end{equation}
Substituting these candidates back into the stationarity equations~\eqref{eq:pda}-\eqref{eq:pds} gives:
\begin{itemize}
\item $r=1/2$ implies $(a,s)=(1/3,2/3)$, thus $1 - a - s=0$, i.e. a boundary point;
\item $r=1$ forces $a=0$, again a boundary point;
\item $r=7/4$ yields no admissible point in the simplex;
\item $r=r_*$ yields no real solution with $a,b,s,e>0$.
\end{itemize}
Therefore, there is no interior critical point.
\\

\noindent
Since the domain is compact, the global maximum of $\Phi$ must lie on the boundary.
By Eq.~\eqref{eq:boundary-bound}, it follows that $\Phi\le \frac{11}{4}$, thus we conclude
\begin{equation}
\phi_{ABC}\le \frac{99}{2} \ .
\end{equation}
Equality holds precisely at
\begin{equation}
a=e=\frac{1}{2}\ ,\quad b=c=d=0\ ,
\end{equation}
namely for the GHZ-state.
\end{proof}
\section{Derivation of the OPE coefficient~\eqref{eq:C123-holographic}}\label{appx:OPE}
\label{app:twist-ope}
\noindent
We derive the OPE coefficient $C_{123}$ for the twist operators in Eq.~\eqref{eq:neg3-3pt} in a holographic CFT$_2$ with central charge~$c$.
We first conformally map the replica surface to a flat torus, then account for the conformal anomaly, and finally evaluate the torus partition function using a semiclassical bulk saddle.
We impose the following two-point normalization
\begin{equation}
\langle\sigma_g(z,\bar{z})\sigma_{g^{-1}}(0)\rangle = |z|^{-2\Delta} = |z|^{-4c/9}\ .
\label{eq:twist-unit-normalization}
\end{equation}
The covering-space construction and its Weyl-anomaly factor are exact; the semiclassical Einstein-gravity approximation is introduced only when evaluating the partition function on the covering torus. 
All three twist operators in Eq.~\eqref{eq:neg3-3pt} correspond to the permutation $g=(123)$ and thus have scaling dimension $\Delta = 2c/9$, as in Eq.~\eqref{eq:twist-dim}.
Conformal symmetry fixes the three-point function up to a constant $C_{123}$:
\begin{equation}
\langle\sigma_g(z_1) \sigma_g(z_2) \sigma_g(z_3)\rangle = \frac{C_{123}}{|z_{12} z_{23} z_{31}|^{2c/9}}\ ,
\label{eq:three-twist-coefficient}
\end{equation}
where $z_{ij} := z_i - z_j$. The Möbius transformation 
\begin{equation}
\xi = \frac{(z-z_1)(z_2-z_3)}{(z-z_3)(z_2-z_1)}
\end{equation}
maps $(z_1, z_2, z_3)$ to $(0, 1, \infty)$.
Defining the operator at infinity by $\sigma_g(\infty) = \lim_{|\xi|\to\infty} |\xi|^{4c/9} \sigma_g(\xi)$, we obtain $C_{123} = \langle\sigma_g(\infty) \sigma_g(1) \sigma_g(0)\rangle$. 
\\

\begin{figure}[htbp]
\centering
\includegraphics[scale = 1.1]{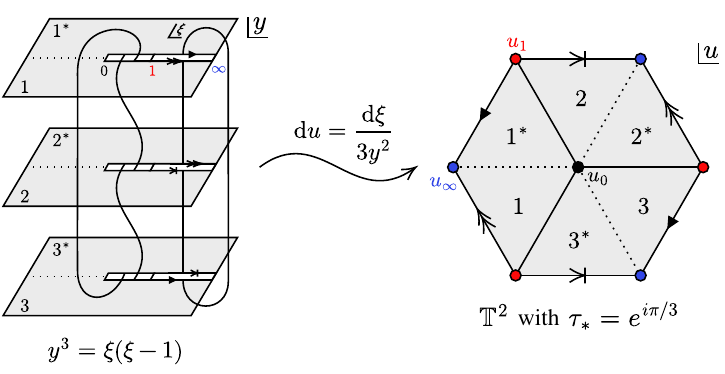}
\caption{Uniformization of the three-sheeted replica surface. Left: the cyclic cover $y^3 = \xi(\xi - 1)$, branched at $\xi = 0, 1, \infty$. The labels $1, 2, 3$ and $1^*, 2^*, 3^*$ denote the ket and bra half-planes, respectively.
Right: each half-plane is mapped to an equilateral triangle in the $u$ plane, with edge length $\ell$ given by Eq.~\eqref{eq:edge-length}.
The six triangles form a regular hexagon with opposite sides identified, giving a torus with modulus $\tau_* = \exp(i\pi/3)$.
The points above $\xi=0,1,\infty$ have coordinates $u_0,u_1,u_\infty$, respectively. }
\label{fig:map}
\end{figure}

\noindent
As shown in Fig.~\ref{fig:map} (left), the replica surface is a connected three-sheeted cover of the $\xi$-sphere. It can be described by the compactification of
\begin{equation}
\Sigma:\ y^3 = \xi(\xi-1)\ ,
\label{eq:explicit-replica-curve}
\end{equation}
which has one ramification point with index three above each of $\xi = 0, 1, \infty$. 
We introduce a flat coordinate $u$ by
\begin{equation}
\mathrm{d} u = \frac{\mathrm{d}\xi}{3y^2}\ .
\label{eq:explicit-flat-coordinate}
\end{equation}
The points above $\xi=0,1,\infty$ have coordinates $u_0,u_1,u_\infty$, respectively.
We choose $u_\infty=0$ to be the origin.
Each bra or ket half-plane is mapped to an equilateral triangle with side length
\begin{equation}\label{eq:edge-length}
\ell = \int_{0}^{1} \Big|\frac{\mathrm{d}\xi}{3y^2}\Big| = \frac{\Gamma(1/3)^3}{2\pi\sqrt{3}}\ ,
\end{equation}
and the six triangles form a regular hexagon with opposite sides identified, giving a flat torus with primitive periods
\begin{equation}
\omega_1 = \sqrt{3}\ell e^{-i\pi/6}\ ,\quad
\omega_2 = \sqrt{3}\ell e^{i\pi/6}\ ,\quad
\tau_* = \frac{\omega_2}{\omega_1} = e^{i\pi/3}\ .
\end{equation}

\noindent
To obtain an explicit expression for the covering map, we set
\begin{equation}
 X=y\ ,\quad Y = 2\xi-1\ .
\end{equation}
Eq.~\eqref{eq:explicit-replica-curve} then becomes $Y^2=4X^3+1$.
Differentiating Eq.~\eqref{eq:explicit-replica-curve} gives
$3y^2\md y=(2\xi-1)\md \xi$.
Together with Eq.~\eqref{eq:explicit-flat-coordinate}, this yields
\begin{equation}\label{eq:diff-eq}
\md u = \frac{\md X}{Y}\ ,\quad \left(\frac{\md X}{\md u}\right)^2 = 4X^3 + 1\ .
\end{equation}
Near $u=u_\infty=0$, the functions behave as $X\sim u^{-2}$ and $Y\sim-2u^{-3}$. 
The differential equation~\eqref{eq:diff-eq} and this local behavior give
\begin{equation}
X=\wp(u|\Lambda)\ ,\quad Y = \wp'(u|\Lambda)\ , \quad \Lambda=\omega_1\mathbb{Z}+\omega_2\mathbb{Z}\ .
\end{equation}
Here $\wp(u|\Lambda)$ is the Weierstrass elliptic function with period lattice $\Lambda$, and the prime denotes differentiation with respect to $u$.
The covering map from the torus $\mathbb{C}/\Lambda$ to the $\xi$-sphere is therefore
\begin{equation}
\xi = \pi(u)=\frac{\wp'(u|\Lambda)+1}{2}\ .
\label{eq:explicit-cover}
\end{equation}
\\

\noindent
We next account for the conformal anomaly associated with changing from the metric pulled back from the $\xi$-sphere to the flat metric on the covering torus.
We follow the regularization procedure in Refs.~\cite{Gadde:2024taa,Harper:2024ker}, normalizing the twist operators as in Eq.~\eqref{eq:twist-unit-normalization}.
To state the general formula, consider a connected genus-one degree-$N$ cover of the $\xi$-sphere.
At each insertion point $\xi_a$, including infinity, we assume the permutation associated with $\sigma_a$ has $m_a$ disjoint cycles, each of length $k_a$.
Thus, $m_a k_a=N$.
Denote the preimages of $\xi_a$ by $u_{a,j}$, $j=1,\ldots,m_a$, and define the local coefficients by
\begin{align}
\pi(u)&=\xi_a+\alpha_{a,j}(u-u_{a,j})^{k_a}+\cdots\ ,\quad a\neq\infty\\
\pi(u)&=\beta_j(u-u_{\infty,j})^{-k_\infty}+\cdots\ .
\label{eq:explicit-general-local-data}
\end{align}
We define the normalized correlator $\mathcal{G} := \big\langle\prod_{a\neq\infty} \sigma_a(\xi_a) \sigma_\infty(\infty) \big\rangle$. 
The conformal anomaly then gives\footnote{We obtain Eq.~\eqref{eq:explicit-general-formula} by summing the local contributions to the Liouville action, following the genus-one prescription in Appendix~A of Ref.~\cite{Gadde:2024taa}, and imposing unit two-point normalization.
All $m_\infty$ preimages of $\xi=\infty$ are located at finite points on the covering torus and are included in the sum over~$\beta_j$.}
\begin{equation}
\begin{split}
\log\frac{\mathcal{G}}{Z_{\mathcal C}(\tau,\bar\tau)} = -\frac c{12}
\Bigg\{&\sum_{a\neq\infty}\bigg[\frac{k_a - 1}{k_a}\sum_{i=1}^{m_a}\log|\alpha_{a,i}| + (N + m_a)\log k_a\bigg] \\
&+ \frac{k_\infty + 1}{k_\infty} \sum_{j=1}^{m_\infty}\log|\beta_j| - (N - m_\infty)\log k_\infty\Bigg\}\ .
\end{split}
\label{eq:explicit-general-formula}
\end{equation}
Here $\tau$ is the modulus of the covering torus, and $Z_{\mathcal{C}}$ is the partition function of a single copy of the original CFT $\mathcal C$ on this torus.
\\

\noindent
For the covering map~\eqref{eq:explicit-cover}, the points $u_0$ and $u_1$ satisfy
\begin{equation}
\wp'(u_0)=-1\ ,\quad \wp'(u_1)=1\ ,\quad \wp(u_0) = \wp(u_1)=0\ .
\end{equation}
Using $\wp''=6\wp^2$ and the Laurent expansion at $u_\infty=0$, we obtain the local expansions
\begin{align}
\pi(u_0+t)&=t^3+\mathcal{O}(t^6)\ ,\\
\pi(u_1+t)-1&=t^3+\mathcal{O}(t^6)\ ,\\
\pi(u_\infty + t)&=-t^{-3}+\frac{1}{2}+\mathcal{O}(t^3)\ .
\end{align}
Since each branch point has a single preimage, we label the local coefficients by the branch points:
\begin{equation}
\alpha_0=\alpha_1=1\ ,\quad\beta_1=-1\ .
\end{equation}
Substituting these coefficients, together with $N=3$, $k_0=k_1=k_\infty=3$, and $m_0=m_1=m_\infty=1$, into Eq.~\eqref{eq:explicit-general-formula} gives
\begin{equation}
\begin{split}
\log\frac{C_{123}}{Z_{\mathcal C}(\tau_*,\bar\tau_*)} &= -\frac c{12}\left[\frac23\log|\alpha_0\alpha_1| + \frac43\log|\beta_1| + 6\log3\right] = -\frac c2\log3\ \\
&\Longrightarrow\ C_{123}=3^{-c/2}Z_{\mathcal C}(\tau_*,\bar\tau_*)\ .
\end{split}
\label{eq:explicit-substitution}
\end{equation}
This equation is exact without requiring large $c$.
\\

\noindent
Finally, we evaluate the covering-torus partition function, at leading order in large $c$, by semiclassical Einstein-gravity solid-torus saddles~\cite{Gadde:2024taa, Harper:2024ker}
\begin{equation}
\begin{split}
S_{\text{min}}(\tau) &= \min_{\mathsf{a},\mathsf{b},\mathsf{c},\mathsf{d}\in\mathbb{Z},\mathsf{a}\mathsf{d} - \mathsf{b}\mathsf{c} = 1}\left[\frac{i\pi c}{12}\left(\frac{\mathsf{a}\tau + \mathsf{b}}{\mathsf{c}\tau + \mathsf{d}} - \frac{\mathsf{a}\bar{\tau} + \mathsf{b}}{\mathsf{c}\bar{\tau} + \mathsf{d}}\right)\right]\ ,\quad\\ 
Z_{\mathcal{C}}(\tau,\bar{\tau}) &= \exp\left[-S_{\text{min}}(\tau) + \mathcal{O}(c^0)\right]\ .
\end{split}
\end{equation}
For $\tau_*=\exp(i\pi/3)$, the dominant saddle gives~\cite{Harper:2024ker}
\begin{equation}
\log Z_{\mathcal C}(\tau_*,\bar\tau_*)
=\frac{\pi\sqrt3}{12}c + \mathcal{O}(c^0)\ .
\label{eq:torus-saddle-result}
\end{equation}
Combining this with Eq.~\eqref{eq:explicit-substitution} yields Eq.~\eqref{eq:C123-holographic}. 
\\

\noindent
The OPE coefficient~\eqref{eq:C123-holographic} obtained above is a special case of a more general formula.
Consider three twist operators $\sigma_a$, whose permutations consist of
$m_a$ disjoint cycles of length $k_a\geq 2$, with $m_a k_a = N$.
Assuming that they define a connected $N$-sheeted cover of the sphere,
the Riemann-Hurwitz formula gives
\begin{equation}
2g_{\Sigma}-2=-2N+\sum_{a=1}^3m_a(k_a-1)
=N\bigg(1-\sum_{a=1}^3\frac{1}{k_a}\bigg)\ . 
\end{equation}
Therefore, the cover has genus one precisely when
\begin{equation}
    \frac{1}{k_1}+\frac{1}{k_2}+\frac{1}{k_3}=1~.
    \label{eq:torus-twist-condition}
\end{equation}
With unit two-point normalization, their OPE coefficient is
\begin{equation}
    C_{123}=Z_{\mathcal C}(\tau,\bar\tau)\,e^P~,
    \quad
    P=-\frac{c}{6}\sum_{a=1}^{3}\frac{N}{k_a}\log k_a~,
    \label{eq:torus-twist-OPE}
\end{equation}
where $c$ is the central charge of the original CFT $\mathcal C$,
and $Z_{\mathcal C}(\tau,\bar\tau)$ is its partition function on the
covering torus, whose modulus $\tau$ is determined by the permutations.
\\

\noindent
We can also compare Eq.~\eqref{eq:torus-twist-OPE} with Eq.~(4.32) in Ref.~\cite{Chang:2016ftb}. Using the replicated central charge $Nc$, the holomorphic weight of the $a$-th twist operator
\begin{equation}
h_a=\frac{cm_a}{24}\bigg(k_a-\frac{1}{k_a}\bigg)=\frac{Nc}{24}\bigg(1-\frac{1}{k_a^2}\bigg)
\end{equation}
matches $h_a=(Nc/6)\eta_a(1-\eta_a)$ with $0\le\eta_a<1/2$, which yields $\eta_a=(1-k_a^{-1})/2$ and $\sum_a\eta_a=1$. 
The exponent $P$ agrees with Eq.~(4.32) in Ref.~\cite{Chang:2016ftb}. 
\\

\noindent
Up to permutations, Eq.~\eqref{eq:torus-twist-condition} has only three integer solutions: $(k_1,k_2,k_3) = (3,3,3)$, $(2,4,4)$, and $(2,3,6)$.
For $(3,3,3)$, $N=3$ reproduces our result~\eqref{eq:explicit-substitution}, while $N=9$ gives the coefficient $C_{123} = 3^{-3c/2}Z_{\mathcal{C}}(\tau,\bar\tau)$ appearing in the third R\'enyi multi-entropy for three adjacent regions computed in Ref.~\cite{Harper:2024ker}.
For $(2,4,4)$, $N=4$ describes the fourth-order negativity $\Tr[(\rho_{BC}^{\Gamma})^4]$ of two adjacent intervals, with $C_{123}=2^{-c}Z_{\mathcal{C}}(i,-i)$.
\bibliographystyle{JHEP}
\bibliography{biblio}
\end{document}